%% file: MAPPD-Dependencies.tex
\newif\ifConference
\newif\ifJournal
\newif\ifAnonymous
\newif\ifArXiv

\ArXivtrue
\Conferencetrue

\documentclass[a4paper,USenglish,cleveref,autoref,thm-restate,numberwithinsect,nameinlink]{lipics-v2021}

\include{commands.tex}

\ifAnonymous

	\author{BLIND}
	{}
	{}
	{}
	{}
	
	\authorrunning{DOUBLE BLIND}

\else
	\author{Mark Jones}
	{TU Delft, The Netherlands}
	{m.e.l.jones@tudelft.nl}
	{https://orcid.org/0000-0002-4091-7089}
	{Partially funded by the Dutch Organisation for Scientific Research (NWO) grant OCENW.KLEIN.125 and OCENW.M.21.306.}
	\author{Jannik Schestag}
	{TU Delft, The Netherlands}
	{j.t.schestag@tudelft.nl}
	{https://orcid.org/0000-0001-7767-2970}
	{Funded by the Dutch Research Council (NWO), project “Optimization for and with Machine Learning (OPTIMAL)” OCENW.GROOT.2019.015.}
	
	\authorrunning{Jones and Schestag}
	\Copyright{Jones and Schestag}

\fi



\usetikzlibrary{shadows}
\tikzstyle{big}=[inner sep=2.5pt]
\tikzstyle{small}=[inner sep=2pt]
\tikzstyle{tiny}=[inner sep=1.7pt]
\tikzstyle{textnode}=[inner sep=0pt]

\tikzstyle{bgvertex}=[circle, minimum size=12pt]
\tikzstyle{vertex}=[circle, draw, fill=white]
\tikzstyle{reti}=[vertex, fill=black]
\tikzstyle{leaf}=[vertex, rectangle]

\tikzstyle{arc}=[arrows = {-Stealth[length=5pt]}]
\tikzstyle{thickarc}=[arrows = {-Stealth[length=8pt]}]

\newcommand{\PDD}{\PROB{\mbox{$\varepsilon$-PDD}}}
\newcommand{\fPDD}{\PROB{\mbox{1-PDD}}}
\newcommand{\WPDD}{\PROB{\mbox{Weighted-PDD}}}

\newcommand{\MAPPD}{\PROB{\mbox{Map-PD}}}
\newcommand{\MAPPDD}{\PROB{\mbox{Map-$\varepsilon$-PDD}}}
\newcommand{\fMAPPDD}{\PROB{\mbox{Map-1-PDD}}}
\newcommand{\WMAPPDD}{\PROB{{Map-Weighted-PDD}}}

\newcommand{\exWMAPPDD}[1]{\PROB{\mbox{ex-{$#1$}-colored-Map-W-PDD}}}
\newcommand{\kSWMAPPDD}{\PROB{\mbox{$k\Hgt$-colored-Map-W-PDD}}}

\newcommand{\viable}{{$\varepsilon$-viable}\xspace}
\newcommand{\fviable}{{$1$-viable}\xspace}
\newcommand{\gviable}{{$\gamma$-viable}\xspace}
\newcommand{\pgviable}{{part-$\gamma$-viable}\xspace}

\newcommand{\GW}{GW}

\newcommand{\Ext}{\textsc{Ext}_\Net}

\newcommand{\todos}[2][]{\todo[#1,color=red!25!green!50]{ #2}}
\newcommand{\todosi}[2][]{\todo[inline,color=red!25!green!50]{ #2}}

\usepackage{hyphenat}
\title{Parameterized Algorithms for Diversity of Networks with Ecological Dependencies}

\titlerunning{Diversity of Networks with Ecological Dependencies}

\begin{document}

\maketitle

\begin{abstract}
	For a phylogenetic tree, the phylogenetic diversity of a set~$A$ of taxa is the total weight of edges on paths to~$A$.
	Finding small sets of maximal diversity is crucial for conservation planning, as it indicates where limited resources can be invested most efficiently.
	In recent years, efficient algorithms have been developed to find sets of taxa that maximize phylogenetic diversity either in a phylogenetic network or in a phylogenetic tree subject to ecological constraints, such as a food web.
	However, these aspects have mostly been studied independently.
	Since both factors are biologically important, it seems natural to consider them together.
	
	In this paper, we introduce decision problems where, given a phylogenetic network, a food web, and integers~$k$, and~$D$,
	the task is to find a set of~$k$ taxa with phylogenetic diversity of at least~$D$ under the~\emph{maximize all paths} measure, while also satisfying viability conditions within the food web.
	Here, we consider different definitions of viability, which all demand that a ``sufficient'' number of prey species survive to support surviving predators.
	
	We investigate the parameterized complexity of these problems and present several fixed-parameter tractable (\FPT) algorithms.
	Specifically, we provide a complete complexity dichotomy characterizing which combinations of parameters---out of the size constraint~$k$, the acceptable diversity loss~$\Dbar$, the scanwidth of the food web~$\sw$, the maximum in-degree~$\delta$ in the network, and the network height~$h$---lead to \Wh{1}-hardness and which admit \FPT algorithms.
	
	Our primary methodological contribution is a novel algorithmic framework for solving phylogenetic diversity problems in networks where dependencies (such as those from a food web) impose an order, using a color coding approach.
\end{abstract}

\section{Introduction}
The Sixth Mass Extinction eliminates species and their genera at an unprecedented rate~\cite{ceballos2023mutilation}, even exceeding rates in previous mass extinction events~\cite{pimm2014biodiversity}.
The situation is severe enough that approximately a quarter of Earth's existing species are at threat~\cite{ferreira2024conservation}.
Inaction now means jeopardizing further parts of the animal tree of life.

\looseness=-1
Because resources are not sufficiently available to preserve all species, scientists developed the \emph{phylogenetic diversity} measure~\cite{FAITH1992} to cover the necessity of making an educated decision on which set of species the limited resources should be invested in.
Given a phylogenetic tree with edge weights in which leaves represent present-day species, the phylogenetic diversity of a set of species~$A$ is the total weight of edges on paths from the root to species in~$A$.
Intuitively, a set of species with larger phylogenetic diversity captures a larger variety of genetic material, and is therefore expected to have larger biodiversity.
Phylogenetic diversity became the most popular measure of the significance of a set of species~\cite{vellend2011measuring}.

Apart from its biological relevance, phylogenetic diversity probably also found favor in the eyes of many scientists for its highly desirable trait of being easy to compute~\cite{FAITH1992}.
Indeed, an optimal solution for maximizing phylogenetic diversity can be found with a greedy algorithm~\cite{Pardi2005,steel} and therefore even large instances can be solved within seconds~\cite{PDinSeconds}.

To model further relevant aspects of conservation planning, more problems were defined, in which a special focus was placed upon capturing varying costs of saving species~\cite{GNAP,pardi07}, finding optimal conservation areas~\cite{bordewich2008,moulton}, considering species' extinction times~\cite{TimePD}, or preserving viable sets of species~\cite{moulton}.
In the latter problem a food web, modeling predator-prey relationships, is given in addition to the phylogenetic tree.
It is asked to find a set~$A$ of species that maximizes the phylogenetic diversity and is viable in the sense that each species in~$A$ is either a source in the ecosystem, or finds prey within~$A$.
Due to concerns that one prey could be insufficient, further definitions of viability have been defined~\cite{WeightedFW}.

The notion of phylogenetic trees has been generalized.
Ancestry of species is in recent years more often modeled with phylogenetic networks, which in contrast to phylogenetic trees also allow hybridization events or horizontal gene transfer~\cite{huson2006application}.
Consequently, generalizations of phylogenetic diversity on networks have been proposed~\cite{bordewichNetworks,MAPPD,van2025average,van2024maximizing,WickeFischer2018}.

As phylogenetic networks represent the connection between species a lot better than phylogenetic trees and considering viability constraints are vital,
it is natural to combine these two aspects.
Yet, to the best of our knowledge, this has not been done so far.
Therefore, we take the step and define problems for the maximization of phylogenetic diversity on networks with different viability definitions.
As it was expected that ``a combination of these concepts [would] result in very hard problems.''~\cite{WeightedFW}, we turn to the toolbox of parameterized complexity to break intractability.

Parameterized complexity is one method to cope with \NP-hardness.
In this, we consider a problem~$\Pi$ and a parameter~$p$ of size~$\kappa$ and say that~$(\Pi,p)$ is \FPT if~$\Pi$ can be solved in~$\Oh^*(f(\kappa))$ time, for some computable function~$f$.
If $\Pi$ is \Wh{1}-hard with respect to~$p$ or even \NP-hard if~$\kappa$ is a constant, then the existence of an \FPT-algorithm is unlikely~\cite{cygan,downeybook}.

We consider three definitions of viability, based on whether a non-source species in the food web needs to have one or all of its prey available, or whether a weighted sum of the available prey must reach a certain threshold.
We provide a complete complexity dichotomy, for the latter two of these definitions, in the sense that for every combination of the following parameters, we show whether the defined problems are \Wh{1}-hard, or can be solved with an \FPT-result: the size constraint~$k$; the acceptable diversity loss~$\Dbar$; the scanwidth of the food web~$\sw$; the maximum in-degree~$\delta$ in the network; and the network height~$h$.
For the other definition of viability, we have a near-complete complexity dichotomy that omits only one of the above combinations.
In particular, we provide \FPT algorithms for the most general version of the problem we define for the parameters~$\Dbar+\sw$ and~$k+\sw+\delta+h$.
For the former, we depend on a notion of \emph{anchors} already used in~\cite{TimePD}, but we improve on their algorithmic idea, in two ways.
First, we use a color coding technique which only requires one color per edge.
Secondly, we consider a non-linear order.
By this, we are able to even provide algorithms for the smaller parameter~$\kbar + \sw$ for the version of the problem on trees.

\section{Preliminaries}
\label{sec:prelims}
\subsection{Definitions}
\ifConference
We use the $\mathcal O^*$-notation which omits factors polynomial in the input size.
\fi
For a positive integer $a$,
by $[a]$ we denote the set $\{1,2,\dots,a\}$, and
by $[a]_0$ the set $\{0\}\cup [a]$.
For functions $f:A\to B$, where~$B$ is a family of sets, we define $f(A') := \bigcup_{a\in A'} f(a)$,
and if $B \subseteq \mathbb{Q}$, we define $f_\Sigma(A') := \sum_{a\in A'} f(a)$, for subsets $A'\subseteq A$.
%
%
We write $\{u,v\}$ for an undirected edge between $u$ and $v$ and~$uv$ for a directed edge from~$u$ to~$v$.
For any graph~$G$, we write~$V(G)$ and~$E(G)$ for the set of vertices or edges, respectively.
For a set of edges~$E'$, we write~$V(E')$ for the vertices with at least one endpoint in an edge of~$E'$.
For a vertex set~$V'\subseteq V(G)$, we let~$G[V']:=(V',\{e\in E(G) \mid V(\{e\}) \subseteq V'\})$ denote the subgraph of~$G$ induced by~$V'$.
We generalize to edge sets~$E' \subseteq E(G)$ by~$G[E'] := G[V(E')]$.

\subparagraph*{Phylogenetic Networks and Phylogenetic Diversity.}
For a given set $X$, a \emph{phylogenetic~$X$-network~$\Net = (V,E,\w)$} is a directed, acyclic graph~$(V,E)$ with \emph{edge-weight}~$\w: E\to \mathbb{N}_{>0}$, in which there is a single vertex with an in-degree of~0, the \emph{root}~$\rho$, and~$X$ is the set of vertices with an in-degree of~1 and an out-degree of~0, called the \emph{leaves}.
All other vertices split up into \emph{tree vertices}, which have an in-degree of~1 and an out-degree of at least~2, and \emph{reticulations}, which have an out-degree of~1 and an in-degree of at least~2.
Edges incoming at reticulations (tree vertices) are \emph{reticulation edges} (\emph{tree edges}).
The set of reticulations, tree vertices, reticulation and tree edges are denoted with~$V_R(\Net)$, $V_T(\Net)$, $E_R(\Net)$, and~$E_T(\Net)$, respectively.
A \emph{phylogenetic~$X$-tree~$\Tree=(V,E,\w)$} is a phylogenetic~$X$-network without reticulations.
The set~$X$ is a set of~\emph{taxa}~(species).
We interchangeably use the words taxon and leaf.
In biological applications, the set $X$ is a set of taxa, and the other vertices of~$\Net$ correspond to biological ancestors of these taxa. An edge $e = uv$ represents direct biological inheritance from $u$ to $v$. The weight~$\w(e)$ describes the phylogenetic distance between the endpoints of~$e$.
As these endpoints correspond to distinct, (possibly extinct) species, we may assume this distance is greater than zero.
Reticulations correspond to species that have direct inheritance from multiple ancestors, such as hybrid species.

For a vertex~$v\in V$, the \emph{offspring} $\off(v)$ of~$v$ is the set of leaves~$x\in X$ for which there is a path from~$v$ to~$x$.
Given a set of taxa $A \subseteq X$,
let $E(A)$ denote the set of edges $uv \in E$ with $\off(v)\cap A \neq \emptyset$.
The \emph{phylogenetic diversity} $\PD(A)$ of $A$ is defined by 
\begin{equation}
	\label{eqn:PDdef}
	\PD(A) := \sum_{e \in E(A)} \w(e).
\end{equation}

In other words, the phylogenetic diversity $\PD(A)$ of a set $A$ of taxa is the sum of the weights of edges that are on a path to a taxon in~$A$.

For phylogenetic trees, this measure is well established~\cite{FAITH1992}.
For phylogenetic networks the search for the most relevant measure is still ongoing, but so far the above definition, also called \emph{All-Paths-PD}, is the measure that is ``the simplest''~\cite{bordewichNetworks,MAPPD,van2025average,van2024maximizing,WickeFischer2018} and is the only measure of phylogenetic diversity considered in this paper.

\subparagraph*{Food Webs.}
A \emph{food web~$\Food=(X,E,\gamma)$ on $X$} for a set of taxa $X$, is a directed acyclic graph~$(X,E)$ with an edge weight function~$\gamma: E \to [0,1]$.

For an edge $xy\in E$, we say that $x$ is \emph{prey} of $y$ and $y$ is a \emph{predator} of $x$.
Thus, edges in~$\Food$ are directed from prey to predator.
The set of prey and predators of $x$ is denoted with $\prey{x}$ and $\predators{x}$, respectively.
The set $\preyE{x}$ is the set of edges incoming at $v$ and $\predatorsE{x}$ is the set of edges outgoing of $v$ in~$\Food$.
Taxa without prey are \emph{sources}.
For the problems considered in this paper, instances in which the food web has several sources can in quadratic time be transformed into one where there is only one source~\cite[Observation~2.3]{PDD}.
Therefore, throughout the paper, we assume that~$\Food$ only has a single source~$s_\Food$.

\looseness=-1
For a given food web~\Food,
a set of taxa~$A\subseteq X$ is \emph{\gviable} 
if~$\gamma_\Sigma(\{ux \mid u \in \prey{x} \cap A\}) \ge 1$
for each non-source~$x\in A$.
That is, the total weight of edges incoming from taxa in~$A$ is at least~$1$~\cite{WeightedFW}.
Given also a set~$E' \subseteq E(\Food)$ of edges,
a set of taxa~$A\subseteq X$ is \emph{$E'$-\pgviable} 
if~$\gamma_\Sigma(\{ ux \in \preyE{x} \mid u\in A \text{ or } ux \in E'\}) \ge 1$ 
for each non-source~$x\in A$.
That is, the total weight of edges incoming from taxa in~$A$ together with incoming edges in $E'$ is at least~$1$.
%


We consider two important special cases of viability.
If $\gamma(e) = 1$ for all $e \in E$, then we say a \gviable set~$A \subseteq X$ is \emph{\viable}.
That is equivalent to saying that~$A$ is \viable, if each non-source in~$A$ can find prey in~$A$.
If $\gamma(ux) = 1/|\prey{x}|$ for every edge in~$\preyE{x}$ for each~$x\in X$, then we say a \gviable set~$A \subseteq X$ is \emph{\fviable}.
That is, $A$ is \fviable, if all prey of each taxon in~$A$ are in~$A$.

\subparagraph*{Problem Definitions.}
In the classical \MPDlong~(\MPD) problem, we are given a phylogenetic tree~$\Tree$, and integers~$k$, and~$D$, and it is asked whether a set of~$k$ taxa with a phylogenetic diversity of at least~$D$ exists~\cite{FAITH1992}.

The problems \MAPPD~\cite{bordewichNetworks}, \PDD~\cite{moulton}, and~\fPDD~\cite{WeightedFW} are generalizations of~\MPD in the following sense.
In \MAPPD, we are given a phylogenetic network instead of a tree and the question stays the same---just with the more general phylogenetic diversity measure.
In \PDD, \fPDD and \WPDD, we are, in addition to the input of \MPD, given a food web and the solution set is required to be \viable, \fviable or \gviable, respectively.

In this paper, we consider the following generalizations of these problems.
\problemdef{\WMAPPDD}{
	A phylogenetic~$X$-network~$\Net$, a food web~$\Food$ on $X$, and integers~$k,D \in \mathbb{N}$.
}{
	Is there an \gviable set~$S\subseteq X$ such that~$|S|\le k$, and~$\PD(S)\ge D$?
}
We call the set~$S$ a \emph{solution} of the instance.
The problems \MAPPDD and \fMAPPDD are defined analogously, where the set~$S$ is required to be~\viable and~\fviable, respectively.

\looseness=-1
We note that in \PDD, \fPDD, and \WPDD, as originally defined, the phylogenetic network is required to be a tree.

Throughout the paper, we adopt the convention that $n$ is the number of taxa~$|X| = |V(\Food)|$
and that $m$ is the number of edges of the food web $|E(\Food)|$.

\subparagraph*{Scanwidth.}
For a directed, acyclic graph~$G=(V,E)$, a rooted, directed tree~$T = (V,E')$ is a \emph{tree extension of~$G$} if for each edge~$uv\in E$ there is a path from~$u$ to~$v$ in~$T$.
We say that an edge~$uv\in E$ \emph{passes over} an edge~$e\in E'$ if the (only) path from~$u$ to~$v$ in~$T$ contains~$e$.
For an edge $uv \in E'$, the set of edges that pass over~$uv$ is~$\GW(v)$ and~$T_\Food^{(v)}$ is the set of vertices that can be reached from~$v$ in~$T_\Food$.
The \emph{scanwidth of a tree extension~$T$} of~$G$ is the maximum number of edges of~$G$ that pass over an edge of~$T$.
The \emph{scanwidth of~$G$} is the minimal scanwidth of any tree extension of~$G$~\cite{berry}.
Computing the tree extension and scanwidth of a directed, acyclic graph is~\NP-hard~\cite{berry} and, when considered for phylogenetic networks, \FPT when parameterized by the level of the network~\cite{holtgrefeComputing,holtgrefeExact}.
%
In this paper, we will use the parameter \emph{scanwidth~$\sw$ of the food web}, and assume we are given a tree extension of $\Food$ of minimum scanwidth.

\subparagraph*{Other Main Parameters.}
Here, we define the main parameters which are used in this paper.

For an instance $\Instance = (\Net,\Food,k,D)$ of \WMAPPDD, we define~$\kbar := |X| - k$ and~$\Dbar := \PD(X)-D = \sum_{e \in E}\w(e) - D$.
Observe that if a set~$S$ of~$k$ taxa with diversity ~$D$ is preserved, then~$\kbar$ taxa are not in~$S$ and a diversity of~$\Dbar$ is lost.
We therefore speak of~$\Dbar$ as the \emph{acceptable loss of diversity}.

We define~$\delta$ to be the \emph{maximum in-degree of a reticulation} of~$\Net$.
We define~$h_R$ and~$h_T$ to be the maximum number of reticulations and tree vertices, respectively, which are in a path in~$\Net$.
Further, we define~$h := h_R + h_T$ to be the \emph{height of the network}.

\subparagraph*{Color Coding.}
In this paper, we use color coding methods.
For an in-depth treatment of color coding, we refer the reader to~\cite[Chapter~5]{cygan} and~\cite{alon}.

\begin{definition}
	\label{def:perfectHashFamily}
	For integers~$n$ and~$k$,
	an~\emph{$(n,k)$-perfect hash family $\mathcal{H}$} is a family of mappings $f: [n] \to [k]$ such that for every subset~$Z$ of~$[n]$ of size at most~$k$, there is an $f \in \mathcal{H}$ which is injective when restricted to~$Z$.
\end{definition}
Each mapping in an~$(n,k)$-perfect hash family can be seen as coloring of a set of~$n$ items into~$k$ colors. 
If we are interested in finding a set of at most $k$ items satisfying some property, then we may assume that each element in the set will have a different color under some coloring. This assumption can make solving a problem easier, and we will use it in the algorithms developed in \Cref{sec:Dbar,sec:k}.
\begin{proposition}[\cite{cygan,Naor1995SplittersAN}]\label{prop:perfectHashFamily}
	For any integers $n, k \geq 1$, a $(n,k)$-perfect hash family of size~$e^k k^{\Oh(\log k)} \cdot \log n$ can be constructed in time $e^k k^{\Oh(\log k)} \cdot n \log n$.
\end{proposition}


\subsection{Related work}
\PDD, defined in~\cite{moulton},
was conjectured to be \NP-hard~\cite{spillner}.
A formal prove appeared in~\cite{faller}.

\begin{theorem}[{\cite[Theorem 5.1]{faller}}]
	\label{thm:faller}
	\PDD is \NP-hard even if the phylogenetic tree has a height of 2 and the food web is an out-tree.
\end{theorem}

\PDD has been analyzed within parameterized complexity~\cite{PDD} and it has been shown that~\PDD is \FPT when parameterized with~$D$, but \Wh{1}-hard with respect to~$\Dbar$~\cite{PDD}.
In~\cite{WeightedFW}, further definitions for viability were given, among them~\gviable and~\fviable.
It has been shown that \fPDD is \Wh{1}-hard with respect to~$k$,~$D$,~$\kbar$, and~$\Dbar$~\cite{fPDD}.
Further, \PDD~\cite{PDD} and \WPDD~\cite{WeightedFW} were analyzed with respect to parameters categorizing the structure of the food web.

\begin{theorem}
	\label{thm:PDD+}
	\begin{enumerate}[(a)]
		\item\label{itm:Dbar-PDD}
			\PDD is \Wh{1}-hard when parameterized with~$\kbar$ or~$\Dbar$, even if the phylogenetic tree is a star~\cite[Proposition~5.1]{PDD}.
		\item\label{itm:Dbar-fPDD}
			\fPDD is \Wh{1}-hard when parameterized with~$\kbar$ or~$\Dbar$, even if the phylogenetic tree is a star~\cite[Theorem~3.3]{fPDD}.
		\item\label{itm:k-fPDD}
			\fPDD is \Wh{1}-hard when parameterized with~$k$ or~$D$, even if the phylogenetic tree is a star~\cite[Theorem~3.2]{fPDD}.
	\end{enumerate}
\end{theorem}

In recent years, the question of how to model phylogenetic diversity in networks best has drawn some attention~\cite{bordewichNetworks,van2025average,van2024maximizing,WickeFischer2018}.
The measure \emph{All-Paths-PD}, as defined in~\cite{bordewichNetworks}, is hereby the easiest to understand and also computationally slightly less challenging than other definitions.
\MAPPD is \FPT when parameterized with~$D$~\cite{MAPPD}, and can be solved in~$\Oh^*(2^{\ret})$~time for~$\ret$ being the number of reticulations~\cite{MAPPD},
in~$\Oh^*(2^{\swwithoutN_{\Net}})$~time for $\swwithoutN_{\Net}$ the scanwidth~\cite{holtgrefePANDA},
and in~$\Oh^*(2^{\Oh(\tw)})$~time for $\tw$ the treewidth of the network~\cite{MAPPD}.

\begin{theorem}[\cite{bordewichNetworks,MAPPD}]
	\label{thm:MAPPD}
	\MAPPD is \Wh{2}-hard, when parameterized with the solution size~$k$, even if every path from the root to a leaf contains exactly one tree vertex and one reticulation.
\end{theorem}

Recently, \MAPPD has been considered in semidirected networks.
In semi directed networks,~\MAPPD can be solved in~$\Oh^*(2^\ell)$ time, where~$\ell$ is the level of the network~\cite{holtgrefePANDA}.

\subparagraph*{Our Contribution.}
We analyze \MAPPDD and \fMAPPDD with respect to the parameters~$k$, $\Dbar$, $\sw$, $\delta$, and~$h$ and show, for any combination of these parameters, whether \fMAPPDD is \Wh{1}-hard, or admits an \FPT algorithm.
For \MAPPDD, we only leave one case as a conjecture and prove all others.
All algorithms we prove for the generalization \WMAPPDD.
Parameterized results for \MAPPDD and \fMAPPDD are displayed in~\Cref{tab:results}.

In \Cref{sec:Dbar}, we prove that~\WMAPPDD is \FPT with respect to~$\Dbar + \sw$.
In \Cref{sec:k}, we show that~\WMAPPDD is \FPT with respect to~$k+\sw+h+\delta$.
If any of these parameters is dropped, \fMAPPDD is \Wh{1}-hard.

The algorithm presented in \Cref{sec:Dbar} is our primary methodological contribution.
In this novel approach, only a single color per edge is used in a color coding algorithm.
This can be used to prove that \PDD and \fPDD are \FPT with respect to~$\kbar + \sw$, see \Cref{cor:Dbar+sw}\ref{cor:Dbar+sw2}.
We expect this to be applicable for similar problems parameterized with~$\kbar$ as well.

\looseness=-1
Due to space restrictions, proofs of theorems and lemmas marked with~$\star$ are partly or fully deferred to the appendix.

\begin{table}[t]
	\resizebox{1.1\columnwidth}{!}{%
	\begin{tabular}{lll|ll|ll|ll}
		$\sw$ & $\delta$ & $h$ & \multicolumn{2}{c|}{parameter alone} & \multicolumn{2}{c|}{$k$ + parameter} & \multicolumn{2}{c}{$\Dbar$ + parameter} \\
		\hline
		\XSolidBrush & \XSolidBrush & \XSolidBrush & \multicolumn{2}{l|}{\emph{no parameter}} & \Wh{2}-h & Thm.~\ref{thm:MAPPD},\cite{bordewichNetworks,MAPPD} & \Wh{1}-h & Thm.~\ref{thm:PDD+}\ref{itm:Dbar-PDD}\&\ref{itm:Dbar-fPDD},\cite{PDD,fPDD} \\
		\XSolidBrush & \XSolidBrush & \checkmark & p-\NP-h & Thm.~\ref{thm:faller},\cite{faller} & \Wh{2}-h & Thm.~\ref{thm:MAPPD},\cite{bordewichNetworks,MAPPD} & \Wh{1}-h & Thm.~\ref{thm:PDD+}\ref{itm:Dbar-PDD}\&\ref{itm:Dbar-fPDD},\cite{PDD,fPDD} \\
		\XSolidBrush & \checkmark & \XSolidBrush & p-\NP-h & Thm.~\ref{thm:faller},\cite{faller} & \Wh{2}-h & Cor.~\ref{cor:binary} & \Wh{1}-h & Thm.~\ref{thm:PDD+}\ref{itm:Dbar-PDD}\&\ref{itm:Dbar-fPDD},\cite{PDD,fPDD} \\
		\XSolidBrush & \checkmark & \checkmark & p-\NP-h & Thm.~\ref{thm:faller},\cite{faller} & \multicolumn{2}{l|}{* \emph{see caption}} & \Wh{1}-h & Thm.~\ref{thm:PDD+}\ref{itm:Dbar-PDD}\&\ref{itm:Dbar-fPDD},\cite{PDD,fPDD} \\
		\checkmark & \XSolidBrush & \XSolidBrush & p-\NP-h & Thm.~\ref{thm:faller},\cite{faller} & \Wh{2}-h & Thm.~\ref{thm:MAPPD},\cite{bordewichNetworks,MAPPD} & \FPT & Cor.~\ref{cor:Dbar+sw}\ref{cor:Dbar+sw1} \\
		\checkmark & \XSolidBrush & \checkmark & p-\NP-h & Thm.~\ref{thm:faller},\cite{faller} & \Wh{2}-h & Thm.~\ref{thm:MAPPD},\cite{bordewichNetworks,MAPPD} & \FPT & Cor.~\ref{cor:Dbar+sw}\ref{cor:Dbar+sw1} \\
		\checkmark & \checkmark & \XSolidBrush & p-\NP-h & Thm.~\ref{thm:faller},\cite{faller} & \Wh{2}-h & Cor.~\ref{cor:binary} & \FPT & Cor.~\ref{cor:Dbar+sw}\ref{cor:Dbar+sw1} \\
		\checkmark & \checkmark & \checkmark & p-\NP-h & Thm.~\ref{thm:faller},\cite{faller} & \FPT & Cor.~\ref{cor:k+sw+S-1PDD} & \FPT & Cor.~\ref{cor:Dbar+sw}\ref{cor:Dbar+sw1}
	\end{tabular}
	}
	\caption{This table shows parameterized complexity results.
	Except for the marked cell, \fMAPPDD and \MAPPDD have the same complexity result.
	All \FPT results are new.
	For any set of taxa~$A$, in polynomial time~$\PD(A)$ and whether~$A$ is \gviable can be computed.
	Consequently, by iterating over all subsets of~$X$, \WMAPPDD is trivially \FPT for~$n = k + \kbar < k + \Dbar$.\\
	Entry *:
	\fMAPPDD is \Wh{1}-hard when parameterized with~$k + \delta + h$ or even~$D + \delta + h$~(Thm.~\ref{thm:PDD+}\ref{itm:k-fPDD})
	while we conjecture \MAPPDD to be \FPT when parameterized with~$k + \delta + h$~(Con.~\ref{con:k+d+h}).
	}
	\label{tab:results}
\end{table}

\subsection{Preliminary Observations}
By \Cref{thm:MAPPD} \MAPPD is \Wh{2}-hard with respect to~$k$ even for a network of a small height~\cite{MAPPD}.
In this hardness reduction,
however, the maximal in-degree of reticulations is big.
It is an easy observation, 
that, in polynomial time, one can replace vertices of a large degree with a stack of vertices, where the newly created edges have negligible weight compared to the original edges of the network.
This works for reticulations and tree vertices alike.
We obtain the following result.
\todo{Give this corollary a full proof in the journal version. JS: I think it is already quite complete. But we can think of it in this sense.}

\begin{corollary}
	\label{cor:binary}
	\MAPPD is \Wh{2}-hard with respect to~$k$ even if the network is binary.
\end{corollary}

If the food web is an out-tree, each taxon~$x\ne s_\Food$ has exactly one prey.
We conclude that the result of~\Cref{thm:faller} also holds for \fPDD.
\begin{corollary}
	\label{cor:faller}
	\fPDD is \NP-hard even if the phylogenetic tree has a height of 2 and the food web is an out-tree.
\end{corollary}

\section{Parameter~${\Dbar + \sw}$}
\label{sec:Dbar}

By \Cref{thm:PDD+}\ref{itm:Dbar-PDD} and \ref{itm:Dbar-fPDD}, \PDD and \fPDD are \Wh{1}-hard when parameterized by~$\Dbar$.
In this section, we prove that~\MAPPDD and \fMAPPDD are \FPT when parameterized by~$\Dbar+\sw$.
Further, we show that when the phylogenetic network is a tree, even an \FPT-running time with respect to the smaller parameter~$\kbar+\sw$ is possible.
To prove these results, we present our main methodological contribution.
This approach uses \emph{anchors}, as introduced in~\cite{TimePD}.
Yet, it is new to use this approach either on phylogenetic networks or for the smaller parameter~$\kbar$.

\newcommand{\thmDbar}[1]{
\begin{theorem}[$\star$]
	#1
	Let a tree extension~$T_\Food$ of the food web with scanwidth~$\sw$ be given.
	\begin{enumerate}[(a)]
		\item\label{thm:Dbar+sw1}\WMAPPDD can be solved in~$\Oh(2^{7.530\cdot \Dbar + \sw + \Oh(\log^2(\Dbar))} \cdot n \cdot |E(\Net)| \cdot \log |E(\Net)|)$ time.
		\item\label{thm:Dbar+sw2}\WPDD can be solved in~$\Oh(2^{15.059\cdot \kbar + \sw + \Oh(\log^2(\kbar))} \cdot n \cdot |E(\Net)| \cdot \log |E(\Net)|)$ time.
	\end{enumerate}
\end{theorem}
}
\thmDbar{\label{thm:Dbar+sw}}

As a consequence of this theorem, we obtain these results.
\begin{corollary}
	\label{cor:Dbar+sw}
	\begin{enumerate}[(a)]
		\item\label{cor:Dbar+sw1}\MAPPDD and \fMAPPDD are \FPT with respect to~$\Dbar + \sw$.
		\item\label{cor:Dbar+sw2}\PDD and \fPDD are \FPT with respect to~$\kbar + \sw$.
	\end{enumerate}
\end{corollary}


For the remainder of this section, fix a tree extension~$T_\Food$ of the food web.
Let $x_1,\dots x_n$ be an ordering of the taxa such that if $x_i$ is a parent of $x_j$ in~$T_\Food$ then $i < j$, and if $x_i$ and $x_j$ share a parent in~$T_\Food$ with $i<j$, then every vertex in the subtree of~$T_\Food$ rooted $x_i$ appears before every vertex in the subtree of~$T_\Food$ rooted $x_j$ in the ordering. 
Such an ordering can be found by taking a depth-first traversal of~$T_\Food$.
For a set of edges $F$ in a directed acyclic graph $G$, we say an edge $e = uv \in F$ is a \emph{highest edge in $F$} if no incoming edge of $u$ is in $F$, and similarly we say $uv \in F$ is a \emph{lowest edge in $F$} if no outgoing edge of $v$ is in $F$.
For a vertex $u$ with outgoing edges $uv$ and $uv'$, we say $uv'$ is a \emph{sibling edge of $uv$}.

The main ideas of our proof is as follows.
Our approach uses color coding techniques, wherein we generate a number of colorings on the edges of $\Net$, and seek a solution under the assumption that a certain set fulfills that each edge in the set has a different color.

Our dynamic programming algorithm will search for a structure that we define below as \emph{perfect triples}.
Roughly speaking, a perfect triple $(A,\chi_1,\chi_2)$ consists of a set of taxa $A$ such that $X\setminus A$ is a solution for our instance of \WMAPPDD, and colorings $\chi_1,\chi_2$ assigning each leaf in $A$ to a set of colors.
Suppose the leaves of $A$ are 'killed' one at a time, in order determined by the depth-first traversal of~$T_\Food$.
Then as each leaf is deleted, a certain set of edges will be 'lost' in the sense that they no longer have any offspring that have not been killed.
For each $x \in A$, the set $\chi_1(x)$ corresponds to the colors of the edges that are lost when $x$ is deleted. 
Furthermore, for each highest edge $e$ that is lost when $x$ is killed, there must be a corresponding sibling edge $e'$ that has not yet been lost (otherwise the parent edge(s) of $e$ would also be lost.
The colors of these sibling edges are represented by the set $\chi_2(x)$.
Note that some of these sibling edges may themselves be lost when a later leaf from $A$ is killed.
We may assume, via standard color-coding techniques, that all the lost edges and sibling edges together are multicolored.

Our dynamic programming algorithm will keep track of the existence of perfect triples satisfying certain properties.
In particular, we store the minimum possible total weight of the set of lost edges for a perfect triple.


\looseness=-1
We now give some formal definitions.

\begin{definition}
Fix an integer~$N$ and a mapping~$c: E_T(\Net) \to [N]$.
For a taxon~$x\in X$ and a set of colors~$C \subseteq [N]$, a set of edges~$F \subseteq E_T(\Net)$ is \emph{$(x,C)$-respecting}, if
\begin{itemize}
	\item $F$ contains the edge incoming at~$x$,
	\item $c(e) \not\in C$ for each~$e\in F$,
	\item for each $uv \in F$ there is a directed path from~$v$ to~$x$ within~$\Net[F]$, and
	\item there is a set~$F'$ of edges, which we call \emph{anchors of~$F$}, such that
	\begin{itemize}
		\item $c(e) \in C$ for each edge~$e \in F'$,
		\item $c(e_1) \ne c(e_2)$ for each pair $e_1,e_2 \in F \cup F'$, and
		\item for each edge~$uv \in F$, exactly one of the following is true
		\begin{enumerate}[(1)]
			\item $F$ contains all edges incoming at~$u$ and~$c(e) \not\in C$ for each edge~$e$ outgoing of~$u$.
			\item $F$ does not contain any edge incoming at~$u$ and there is an edge~$e$ outgoing of~$u$ with~$c(e) \in C$ and~$e\in F'$.
		\end{enumerate}
	\end{itemize}
\end{itemize}
When $(x,C)$ is clear from the context we say an $(x,C)$-respecting set is simply \emph{respecting}.
\end{definition}

In~\Cref{fig:exampleRespecting} an example of a respecting set of edges is given and some example networks for which~$x$ under a certain coloring has no respecting set of edges.

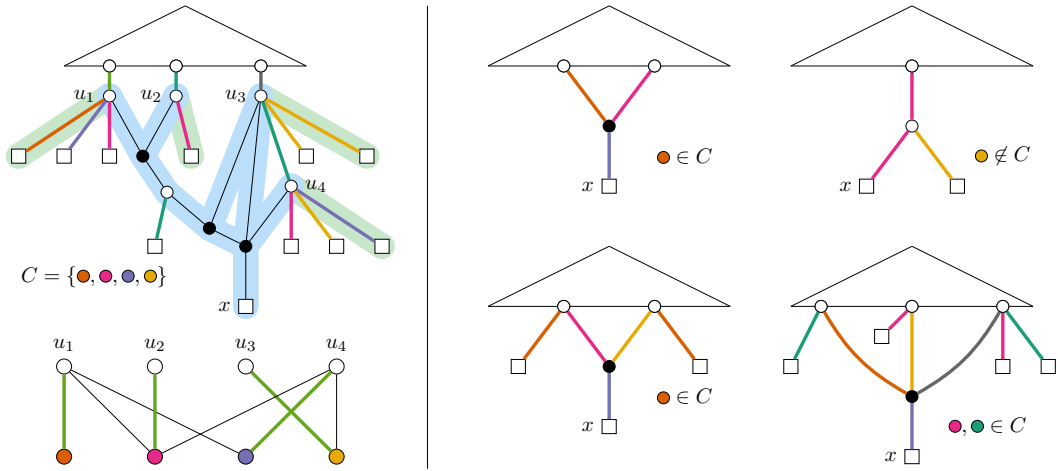
\begin{figure}
	\DeclareRobustCommand{\tikzdot}[1]{\tikz[baseline=-0.6ex]{\node[draw,fill=#1,inner sep=2pt,circle] at (0,0) {};}}
	\definecolor{col1}{HTML}{d95f02}
	\definecolor{col2}{HTML}{e7298a}
	\definecolor{col3}{HTML}{7570b3}
	\definecolor{col4}{HTML}{e6ab02}
	\definecolor{col5}{HTML}{66a61e}
	\definecolor{col6}{HTML}{1b9e77}
	\definecolor{col7}{HTML}{666666}
	
	\definecolor{bgcol1}{HTML}{ebebeb}
	\definecolor{bgcol2}{HTML}{c8e6c9}
	\definecolor{bgcol3}{HTML}{bbdefb}
	
	\centering
	\resizebox{\textwidth}{!}{
		\begin{tikzpicture}
			\node [leaf,label=left:$x$] (0) at (5, -0.5) {};
			\node [textnode] at (6.25, 0) {$\tikzdot{col1}\in C$};
			\node [leaf,label=left:$x$] (1) at (9.25, -0.5) {};
			\node [textnode] at (11.5, 0) {$\tikzdot{col4}\not\in C$};
			\node [leaf] (2) at (10.75, -0.5) {};
			\node [reti,small] (3) at (5, 0.5) {};
			\node [vertex,small] (4) at (4.25, 1.5) {};
			\node [vertex,small] (5) at (5.75, 1.5) {};
			\node [vertex,small] (6) at (10, 0.5) {};
			\node [vertex,small] (7) at (10, 1.5) {};
			\node [vertex,small] (8) at (4.25, -2.5) {};
			\node [vertex,small] (9) at (5.75, -2.5) {};
			\node [reti,small] (10) at (5, -3.5) {};
			\node [leaf] (11) at (3.5, -3.5) {};
			\node [leaf] (12) at (6.5, -3.5) {};
			\node [leaf,label=left:$x$] (13) at (5, -4.5) {};
			\node [textnode] at (6.25, -4) {$\tikzdot{col1}\in C$};
			\node [reti,small] (14) at (10, -4) {};
			\node [leaf,label=left:$x$] (15) at (10, -5) {};
			\node [textnode] at (11.25, -4.5) {$\tikzdot{col2},\tikzdot{col6}\in C$};
			\node [vertex,small] (16) at (8.5, -2.5) {};
			\node [vertex,small] (17) at (10, -2.5) {};
			\node [vertex,small] (18) at (11.5, -2.5) {};
			\node [leaf] (19) at (8, -3.5) {};
			\node [leaf] (20) at (11.5, -3.5) {};
			\node [leaf] (21) at (12.25, -3.5) {};
			\node [leaf] (22) at (9.5, -3) {};
			\coordinate (23) at (10, -1.5) {};
			\coordinate (24) at (8, -2.5) {};
			\coordinate (25) at (12, -2.5) {};
			\coordinate (26) at (5, -1.5) {};
			\coordinate (27) at (3, -2.5) {};
			\coordinate (28) at (7, -2.5) {};
			\coordinate (29) at (10, 2.5) {};
			\coordinate (30) at (8, 1.5) {};
			\coordinate (31) at (12, 1.5) {};
			\coordinate (32) at (5, 2.5) {};
			\coordinate (33) at (3, 1.5) {};
			\coordinate (34) at (7, 1.5) {};

			\node [textnode] at (-3.5, -2) {$C = \{\tikzdot{col1},\tikzdot{col2},\tikzdot{col3},\tikzdot{col4}\}$};
			
			\node [bgvertex,fill=bgcol3] (35) at (-0.75, 1) {};
			\node [bgvertex,fill=bgcol3] (36) at (-0.25, -0.5) {};
			\node [bgvertex,fill=bgcol3] (37) at (-1, -1.5) {};
			\node [bgvertex,fill=bgcol3] (37') at (-1.6, -1.2) {};
			\node [bgvertex,fill=bgcol3] (38) at (-2.7, 0) {};
			\node [bgvertex,fill=bgcol3] (38') at (-2.3, -.6) {};
			\node [bgvertex,fill=bgcol3] (39) at (-2.15, 1) {};
			\node [bgvertex,fill=bgcol3] (40) at (-3.25, 1) {};
			\node [bgvertex,fill=bgcol3] (41) at (-1, -2.5) {};
			\node [bgvertex,fill=bgcol2] (47) at (-4.75, 0) {};
			\node [bgvertex,fill=bgcol2] (42) at (-1.9, 0) {};
			\node [bgvertex,fill=bgcol2] (65) at (1, 0) {};
			\node [bgvertex,fill=bgcol2] (48) at (1.25, -1.5) {};
			\draw [line width=12pt,color=bgcol2] (39.center) to (42.center);
			\draw [line width=12pt,color=bgcol2] (35.center) to (65.center);
			\draw [line width=12pt,color=bgcol2] (40.center) to (47.center);
			\draw [line width=12pt,color=bgcol2] (36.center) to (48.center);
			\node [bgvertex,fill=bgcol3] (35) at (-0.75, 1) {};
			\node [bgvertex,fill=bgcol3] (36) at (-0.25, -0.5) {};
			\node [bgvertex,fill=bgcol3] (39) at (-2.15, 1) {};
			\node [bgvertex,fill=bgcol3] (40) at (-3.25, 1) {};
			\draw [line width=12pt,color=bgcol3] (40.center) to (38.center);
			\draw [line width=12pt,color=bgcol3] (39.center) to (38.center);
			\draw [line width=12pt,color=bgcol3] (35.center) to (37.center);
			\draw [line width=12pt,color=bgcol3] (36.center) to (37.center);
			\draw [line width=12pt,color=bgcol3] (37.center) to (41.center);
			\draw [line width=12pt,color=bgcol3] (38.center) to (38'.center);
			\draw [line width=12pt,color=bgcol3] (38'.center) to (37'.center);
			\draw [line width=12pt,color=bgcol3] (35.center) to (37'.center);
			\draw [line width=12pt,color=bgcol3] (37'.center) to (37.center);

			\node [vertex,small,label=left:$u_3$] (35) at (-0.75, 1) {};
			\node [vertex,small,label=right:$u_4$] (36) at (-0.25, -0.5) {};
			\node [reti,small] (37) at (-1, -1.5) {};
			\node [reti,small] (37') at (-1.6, -1.2) {};
			\node [reti,small] (38) at (-2.7, 0) {};
			\node [vertex,small] (38') at (-2.3, -.6) {};
			\node [vertex,small,label=left:$u_2$] (39) at (-2.15, 1) {};
			\node [vertex,small,label=left:$u_1$] (40) at (-3.25, 1) {};
			\node [leaf,label=left:$x$] (41) at (-1, -2.5) {};
			\node [leaf] (47) at (-4.75, 0) {};
			\node [leaf] (43) at (-4, 0) {};
			\node [leaf] (46) at (-3.25, 0) {};
			\node [leaf] (42) at (-1.9, 0) {};
			\node [leaf] (45) at (0, 0) {};
			\node [leaf] (65) at (1, 0) {};
			\node [leaf] (38'') at (-2.5, -1.5) {};
			\node [leaf] (66) at (-0.25, -1.5) {};
			\node [leaf] (44) at (0.5, -1.5) {};
			\node [leaf] (48) at (1.25, -1.5) {};
			\node [vertex,small] (49) at (-3.25, 1.5) {};
			\node [vertex,small] (50) at (-2.15, 1.5) {};
			\node [vertex,small] (51) at (-0.75, 1.5) {};
			\coordinate (52) at (-2, 2.5) {};
			\coordinate (53) at (-4, 1.5) {};
			\coordinate (54) at (0, 1.5) {};
			\coordinate (55) at (2, 2.5) {};
			\coordinate (56) at (2, -5.2) {};
			
			\node [vertex,big,label=above:$u_1$] (57) at (-4, -3.5) {};
			\node [vertex,big,label=above:$u_2$] (58) at (-2.5, -3.5) {};
			\node [vertex,big,label=above:$u_3$] (59) at (-1, -3.5) {};
			\node [vertex,big,label=above:$u_4$] (60) at (0.5, -3.5) {};
			\node [vertex,big,fill=col1] (61) at (-4, -5) {};
			\node [vertex,big,fill=col2] (62) at (-2.5, -5) {};
			\node [vertex,big,fill=col3] (63) at (-1, -5) {};
			\node [vertex,big,fill=col4] (64) at (0.5, -5) {};

			\draw [col1,ultra thick] (4) to (3);
			\draw [col2,ultra thick] (5) to (3);
			\draw [col3,ultra thick] (3) to (0);
			
			\draw [col2,ultra thick] (7) to (6);
			\draw [col2,ultra thick] (6) to (1);
			\draw [col4,ultra thick] (6) to (2);
			
			\draw [col1,ultra thick] (8) to (11);
			\draw [col2,ultra thick] (8) to (10);
			\draw [col3,ultra thick] (10) to (13);
			\draw [col4,ultra thick] (9) to (10);
			\draw [col1,ultra thick] (9) to (12);
			
			\draw [col6,ultra thick] (16) to (19);
			\draw [col1,ultra thick,bend right=15] (16) to (14);
			\draw [col3,ultra thick] (14) to (15);
			\draw [col2,ultra thick] (17) to (22);
			\draw [col4,ultra thick] (17) to (14);
			\draw [col7,ultra thick,bend left=15] (18) to (14);
			\draw [col2,ultra thick] (18) to (20);
			\draw [col6,ultra thick] (18) to (21);

			\draw (23) to (25);
			\draw (25) to (24);
			\draw (24) to (23);
			\draw (26) to (28);
			\draw (28) to (27);
			\draw (27) to (26);
			\draw (29) to (31);
			\draw (31) to (30);
			\draw (30) to (29);
			\draw (32) to (34);
			\draw (34) to (33);
			\draw (33) to (32);
			
			\draw (40) to (38);
			\draw (39) to (38);
			\draw (35) to (37);
			\draw (36) to (37);
			\draw (37) to (41);
			\draw (37') to (37);
			\draw (35) to (37');
			\draw (38') to (37');
			\draw (38) to (38');
			\draw [col6,ultra thick] (38') to (38'');
			\draw [col2,ultra thick] (39) to (42);
			\draw (35) to (36);
			\draw [col4,ultra thick] (35) to (45);
			\draw [col4,ultra thick] (35) to (65);
			\draw [col1,ultra thick] (40) to (47);
			\draw [col2,ultra thick] (40) to (46);
			\draw [col3,ultra thick] (40) to (43);
			\draw [col2,ultra thick] (36) to (66);
			\draw [col3,ultra thick] (36) to (48);
			\draw [col4,ultra thick] (36) to (44);
			\draw [col5,ultra thick] (49) to (40);
			\draw [col6,ultra thick] (50) to (39);
			\draw [col7,ultra thick] (51) to (35);
			\draw (52) to (54);
			\draw (54) to (53);
			\draw (53) to (52);
			\draw (55) to (56);
			\draw [col6,ultra thick] (35) to (36);
			
			\draw [col5,ultra thick] (57) to (61);
			\draw (57) to (63);
			\draw (57) to (62);
			\draw [col5,ultra thick] (58) to (62);
			\draw [col5,ultra thick] (59) to (64);
			\draw [col5,ultra thick] (60) to (63);
			\draw (60) to (64);
			\draw (60) to (62);

			\node [vertex,small] at (49) {};
			\node [vertex,small] at (50) {};
			\node [vertex,small] at (51) {};
			\node [vertex,small] at (16) {};
			\node [vertex,small] at (17) {};
			\node [vertex,small] at (18) {};
			\node [vertex,small] at (4) {};
			\node [vertex,small] at (5) {};
			\node [vertex,small] at (7) {};
			\node [vertex,small] at (8) {};
			\node [vertex,small] at (9) {};
		\end{tikzpicture}
	}
	\caption{\looseness=-1
		Top Left:
		An example network. The set~$F_{x,C}$ is highlighted in~\tikzdot{bgcol3} and one possible set~$F_{x,C}'$ in~\tikzdot{bgcol2}.
		For the sake of readability, colors of edges in~$F_{x,C}$ are omitted.
		Bottom Left:
		The bipartite graph which is constructed in~\Cref{lem:FxC-compute}.
		A perfect matching is highlighted.\\
		Right:
		Four examples of networks, where~$F_{x,C}$ is not a respecting set of edges.
	}
	\label{fig:exampleRespecting}
\end{figure}

We continue with proving some essential properties about respecting sets.
\begin{lemma}
	\label{lem:FxC-unique}
	For each taxon~$x\in X$ and each set~$C \subseteq [N]$ of colors, at most one set of edges is $(x,C)$-respecting.
\end{lemma}
\begin{proof}
	Towards a contradiction, assume that~$F_1$ and~$F_2$ are $(x,C)$-respecting and that there is~$u'v'\in F_1$ such that~$u'v'\notin F_2$.
	Fix a set $F_1'$ of anchors of $F$, and a set $F_2'$ of anchors of $F_2$.
	Choose any path from~$v'$ to~$x$ in~$\Net[F_1]$ and let~$uv$ be the first edge on this path that occurs in~$F_2$, also.
	Such an edge exists as the edge incoming at~$x$ is in any respecting set of edges.
	As~$F_2$ does not contain all edges incoming at~$u$, there is an edge~$uw \in F_2'$ with~$c(uw) \in C$.
	But then~$F_1$ is not respecting, as $F_1$ contains an edge incoming at $u$.
	%
\end{proof}

As a consequence of~\Cref{lem:FxC-unique}, we write~$F_{x,C}$ for the unique set of $(x,C)$-respecting edges, if existent.
If the set~$F_{x,C}$ exists, we write~$c(F_{x,C})$ for the colors on edges in~$F_{x,C}$.
By definition,~$c(F_{x,C})$ and~$C$ are disjoint.
If the set~$F_{x,C}$ does not exist, we define~$\w(F_{x,C}) = \infty$.
We note that the set of anchors is not necessarily unique.

\begin{lemma}
	\label{lem:FxC-compute}
	For a taxon~$x\in X$ and a set of colors~$C \subseteq [N]$, in~$\Oh(|E(\Net)| \cdot \sqrt{N})$ time, we can compute~$F_{x,C}$, or conclude that it does not exist.
\end{lemma}
\begin{proof}
	We construct a unique set of edges $F$ such that either $F$ is respecting for~$x$ and~$C$, or such a set does not exist, as follows.
	Initially let $F$ contain the unique incoming edge $e$ at $x$.
	Now, for each edge $uv$ in $F$ in turn, if $u$ has no outgoing edge $uv'$ with $c(uv') \in C$, then add all incoming edges of $u$ to $F$.
	Repeat this process exhaustively to complete the construction of $F$.
	If $c(e) \notin C$ for some~$e\in F$ or if two edges have the same color, no respecting set exists.
	By using appropriate data structures, this can be implemented in~$\Oh(|E(\Net)|)$ time.
	
	It remains to decide whether there exists a valid set of anchors $F'$ for $F$.
	In particular, for each highest edge $uv \in F$, we need to choose one outgoing edge $uv'$ with $c(uv') \in C$ to add to $F'$ such that $c(e_1) \neq c(e_2)$ for each pair $e_1,e_2 \in F'$.
	This can be done by reducing to an instance of \textsc{Perfect Matching}, as follows.
	Construct a bipartite graph $G$ with vertex set $F_h \cup C$, where $F_h$ corresponds to the set of highest edges in $F$.
	For each $e \in F_h$ and $c \in C$, add an edge $\{e,c\}$ to $G$ if $e$ has a sibling edge $e'$ with $c(e') = c$.
	Now, $F$ is an $(x,C)$-respecting set if and only if $G$ has a perfect matching covering $F_h$. 
	\textsc{Perfect Matching} can be solved in~$\Oh(|E(G)|\cdot \sqrt{|V(G)|})$ time with the famous Hopcroft-Karp Algorithm~\cite{karpinski1998fast}.
	For each edge in~$\Net$, there is at most one edge in~$G$.
	Thus, the overall running time is~$\Oh(|E(\Net)| \cdot \sqrt{N})$.
\end{proof}

With respecting sets defined, we now formally define perfect triples.
%
A triple~$(A,\chi_1,\chi_2)$, consisting of a set of taxa~$A\subseteq X$, 
and mappings~$\chi_1,\chi_2: A \to 2^{[N]}$, is \emph{perfect},~if
\begin{itemize}
	\item the sets~$\chi_i(x)$ and~$\chi_i(y)$ are pairwise disjoint for~$x,y \in A$ and~$i\in\{1,2\}$,
	\item the sets~$\chi_1(x_i)$ and~$\chi_2(x_j)$ are pairwise disjoint for~$x_i,x_j \in A$ and~$i\le j$, and
	\item for each~$x\in A$ there is a set of respecting edges~$F_{x,\chi_2(x)}\subseteq E(\Net)$ with~$\chi_1(x) = c(F_{x,\chi_2(x)})$.
\end{itemize}

To give some intuition behind the notion of a perfect triple, we observe that the existence of a perfect triple $(A,\chi_1,\chi_2)$ provides a lower bound on the phylogenetic diversity of $(X \setminus A)$.
The key idea is that as we remove the elements of $A$ from $X$, one at a time, the set of edges lost with the removal of each $x\in A$ is a subset of the edges in $F_{x,\chi_2(x)}$.

\newcommand{\lemPerfectTripleBound}[1]{
\begin{lemma}[$\star$]
	{#1}
	$\PD(X\setminus A) \geq \w(E(\Net)) - \sum_{x\in A} \w(F_{x,\chi_2(x)})$
	for every perfect\lb triple~$(A,\chi_1,\chi_2)$.
\end{lemma}
}
\lemPerfectTripleBound{\label{lem:perfectTripleBound}}
\newcommand{\ProofPerfectTripleBound}{
\begin{proof}
	Recall that $x_1,\dots, x_n$ is the ordering of $X$ given by a depth-first traversal of $T_\Food$.
	For the sake of notational convenience, assume that $A = \{x_1,\dots, x_{|A|}\}$.

	The intuitive idea behind our proof is to show that, as we kill each of the taxa $x_1, \dots, x_{|A|}$ in order, the amount of diversity we lose by killing $x_i$ is at most $\w(F_{x_i,\chi_2(x_i)})$.
	To this end, define $\Ext(A') := E(\Net)\setminus E(X\setminus A')$ for any $A' \subseteq X$. That is, $\Ext(A')$ is the set of edges in $\Net$ with all offspring in $A'$.
	
	We prove  the following claim by induction.
	$\Ext(\{x_1,\dots, x_s\}) \subseteq \bigcup_{i = 1}^{s} \w(F_{x_i,\chi_2(x_i)})$, for each $s \in [|A|]$. Note that by letting $s = |A|$, this implies $\PD(X\setminus A) = \w_\Sigma(E(\Net)\setminus \Ext(A)) =  \w_\Sigma(E(\Net) )- \w_\Sigma(\Ext(A)) \geq  \w_\Sigma(E(\Net) - \sum_{x\in A} \w(F_{x,\chi_2(x)})$, as required.
	
	For the base case, observe that $\Ext(\{x_1\})$ consists of the edges $uv$ for which $v$ has a path to $x$, and either~$v=x$ or~$v$ is a reticulation.
	But the definition of a respecting set implies that all such edges are also in $F_{x_1, \chi_2(x_1)}$, and so $\Ext(\{x_1\}) \subseteq F_{x_1, \chi_2(x_1)}$.
	
	For the inductive step, assume that $\Ext(\{x_1,\dots, x_s\}) \subseteq \bigcup_{i = 1}^{s} \w(F_{x_i,\chi_2(x_i)})$ for some $s < |A|$, and now assume for a contradiction that $\Ext(\{x_1,\dots, x_{s+1}\})$ is not a subset of  $\bigcup_{i = 1}^{s+1} \w(F_{x_i,\chi_2(x_i)})$.
	Then consider a lowest edge $e = uv$ in  $\Ext(\{x_1,\dots, x_{s+1}\}) \setminus \bigcup_{i = 1}^{s+1} \w(F_{x_i,\chi_2(x_i)})$.
	By definition edges outgoing of $v$ are also in $\Ext(\{x_1,\dots, x_{s+1}\})$, and as $uv$ is a lowest edge, all edges outgoing of $v$ are also in $\bigcup_{i = 1}^{s+1} \w(F_{x_i,\chi_2(x_i)})$.
	At least one child edge $vw$ is not in $\Ext(\{x_1,\dots, x_{s}\})$ (otherwise $uv \in \Ext(\{x_1,\dots, x_{s}\}) \subseteq  \bigcup_{i = 1}^{s} \w(F_{x_i,\chi_2(x_i)})$, and this edge $vw$ must be in $F_{x_{s+1}, \chi_2(x_{s+1})}$.
	Then as $v$ has an edge incoming which is not in $F_{x_{s+1}, \chi_2(x_{s+1})}$, it follows by definition of a respecting set that $vw$ has a sibling edge $vw'$ with $c(vw') \in \chi_2(s+1)$.
	
	Note however that $vw' \notin  \bigcup_{i = 1}^{s+1} \w(F_{x_i,\chi_2(x_i)})$, as $c(vw') \in \chi_2(x_{s+1})$, which is disjoint from $\chi_1(v_i)$ for each $i \leq s$, and $\chi_1(e) \in \chi_1(v_i)$ for each $e \in F_{x_i,\chi_2(x_i)}$.
	Therefore, $vw'$ is also not in $\Ext(\{x_1,\dots, x_{s+1}\})$, as otherwise $uv$ was not a lowest edge in  $\Ext(\{x_1,\dots, x_{s+1}\}) \setminus \bigcup_{i = 1}^{s+1} \w(F_{x_i,\chi_2(x_i)})$.
	It follows that $w'$, and therefore $v$, has a path to a leaf which is not in $\{x_1,\dots, x_{s+1\}}$, contradicting the assumption that $uv$ was not a lowest edge in  $\Ext(\{x_1,\dots, x_{s+1}\})$.
\end{proof}
}

Having these definitions, we now define a colored problem, which we use as an auxiliary problem for solving~\Cref{thm:Dbar+sw}.
In \exWMAPPDD{N}, besides the usual input of \WMAPPDD
\ifConference
$(\Net,\Food,k,D)$---where~$\gamma$ is a weighting of~$\Food$---%
\fi
we are given a mapping~$c: E(\Net) \to [N]$ of a color per edge.
We ask whether there exists a set of taxa~$A\subseteq X$ of size \emph{at least}~$\kbar$
and mappings~$\chi_1,\chi_2: A \to 2^{[N]}$  such that~$X\setminus A$ is~\gviable, 
the triple~$(A,\chi_1,\chi_2)$ is perfect, and~$\sum_{x\in A} \w(F_{x,\chi_2(x)}) \le \Dbar$.

Note that the existence of a perfect triple $(A,\chi_1,\chi_2)$ does not imply that $X\setminus A$ is \gviable.

\newcommand{\lemexPPD}[1]{
\begin{lemma}[$\star$]
	{#1}
	Given a tree-extension~$T_\Food$ of the food web, we can solve instances of \exWMAPPDD{N} in~$\Oh(5^N \cdot 2^{\sw} \cdot n \cdot (N \cdot \kbar + |E(\Net)|))$~time.
\end{lemma}
}
\newcommand{\IntuitionexPPD}{
	We consider the tree extension of the food web with a dynamic program bottom up.
	At each vertex $v$, we determine whether there exists a perfect triple $(A,\chi_1,\chi_2)$ satisfying certain conditions, where $A$ is a subset of $T_\Food^{(v)}$, the set of vertices descended from $v$ in $T_\Food$.
	We want that $X\setminus A$ is \gviable; to help determine this we keep track of a subset of edges $\Phi \subseteq  \GW(v)$, and require that $T_\Food^{(v)} \setminus A$ is $\Phi$-\pgviable.
}
\lemexPPD{\label{lem:exPPD}}
The intuition behind~\Cref{lem:exPPD} is as follows:
\IntuitionexPPD
\newcommand{\ProofexPPD}{
\begin{proof}
	\proofpara{Intuition}
	\IntuitionexPPD

	\proofpara{Table Definition}
	For a set of taxa~$Q\subseteq X$, sets of colors~$C_1,C_2\subseteq [N]$, a set of edges~$\Phi \subseteq E(\Food)$ between~$Q$ and~$X\setminus Q$, and an integer~$\ell\in [\kbar]_0$,
	let~$\SSS_{(Q,C_1,C_2,\Phi,\ell)}$ be the set of perfect 
	triples~$(A,\chi_1,\chi_2)$ of a set~$A \subseteq Q$ of size \emph{at least}~$\ell$ and mappings~$\chi_1,\chi_2: A \to 2^{C}$ with~$\chi_1(A) = C_1$ and~$\chi_2(A)\setminus \chi_1(A) \subseteq C_2$, for which~$Q\setminus A$ is~$\Phi$-\pgviable.
	
	We define a dynamic programming algorithm with table~$\DP$.
	For a vertex~$v\in V(T_\Food)$, sets of colors~$C_1,C_2\subseteq [N]$, a set of edges~$\Phi \subseteq \GW(v)$, and an integer~$\ell\in [k]_0$, we store in~$\DP[v,C_1,C_2,\Phi,\ell]$ the \emph{minimum} value of~$\w(F_{A,\chi_2}) := \sum_{x\in A} \w(F_{x,\chi_2(x)})$ for any perfect triple~$(A,\chi_1,\chi_2) \in \SSS_{(T_\Food^{(v)},C_1,C_2,\Phi,\ell)}$.
	%
	%
	If~$v$ an internal vertex of~$T_\Food$ with children~$w_1,\dots,w_t$,
	then we define further auxiliary tables~$\DP_i[v,C_1, C_2,\Phi,\Psi,\ell]$ analogously, with~$(S,\chi_1,\chi_2)$ considered from~$\SSS_{(Q_i,C_1,C_2,P_i,\ell)}$ with~$Q_i := \bigcup_{j=1}^i T_\Food^{(w_j)}$ and~$P_i := (\Phi \cup \Psi) \cap \left(\bigcup_{j=1}^i \GW(w_j)\right)$;
	that is---only the first~$i$ children of~$v$ are considered.
	We require that the set of edges~$\Psi$ is either empty or~$\predatorsE{v}$.
	
	\proofpara{Algorithm}
	We define the table in a bottom-up fashion.
	Let~$v$ be a leaf of~$T_\Food$.
	(That is a taxon without predators).
	Fix disjoint color sets~$C_1,C_2 \subseteq [N]$.
	If~$\ell > 1$, we store~$\infty$ in~$\DP[v,C_1,C_2,\Phi,\ell]$.
	If~$\ell=0$ and~$\gamma_\Sigma(\Phi) \ge 1$, we store~0 in~$\DP[v,C_1,C_2,\Phi,\ell=0]$.
	Otherwise---if~$\ell=1$ or~$\gamma_\Sigma(\Phi) < 1$---we store~$\min \{ \w(F_{v,C_2'}) \mid C_2' \subseteq C_2, c(F_{v,C_2'}) = C_1\}$ in~$\DP[v,C_1,C_2,\Phi,\ell=1]$.
	
	Now, let~$v$ be an internal vertex of~$T_\Food$ with children~$w_1,\dots,w_t$ such that $w_i$ comes before $w_{i+1}$ in the depth-first traversal ordering of $X$, for $i \in [t-1]$.
	We define~$\DP_1[v,C_1,C_2,\Phi,\Psi,\ell]$ to be~$\DP[w_1,C_1,C_2,(\Phi \cup \Psi) \cap \GW(w_1),\ell]$.
	To compute further values for~$i\in [t-1]$, we use the following recurrence in which we define~$\DP_{i+1}[v,C_1,C_2,\Phi,\Psi,\ell]$ to be
	\begin{eqnarray}
		\label{rec:D-i}
		&\min&
		\DP_i[v,C_1',C_2' \cup (C_1\setminus C_1'),\Phi,\Psi,\ell']\\
		&& +
		\DP[w_{i+1},C_1 \setminus C_1',C_2 \setminus C_2',(\Phi \cup \Psi) \cap \GW(w_{i+1}),\ell-\ell'].
	\end{eqnarray}
	Here, we take the minimum over all~$C_1' \subseteq C_1, C_2' \subseteq C_2$ and $\ell'\in [\ell]_0$.
	
	Finally, if~$v=s_\Food$ or $\gamma_\Sigma(\Phi \cap \preyE{v}) \ge 1$, we set~$\DP[v,C_1,C_2,\Phi,\ell]$ to be
	\begin{eqnarray}
		\label{rec:D-t}
		\min\{~\DP_t[v,C_1,C_2,\Phi,\predatorsE{v},\ell]~;~\DP_t'~\}.
	\end{eqnarray}
	Here,~$\DP_t'$ is~$\min_{C_2' \subseteq C_2}~\DP_t[v,C_1 \setminus c(F_{v,C_2'}),C_2 \setminus C_2',\Phi,\emptyset,\ell-1] + \w(F_{v,C_2'})$.
	Otherwise, we set~$\DP[v,C_1,C_2,\Phi,\ell]$ to be~$\DP_t'$.
	Intuitively, $\DP_t'$ corresponds to the case that $v$ is one of the taxa to be going extinct, while $\DP_t[v,C_1,C_2,\Phi,\predatorsE{v},\ell]$ corresponds to the case that $\DP_t[v,C_1,C_2,\Phi,\predatorsE{v},\ell]$ corresponds to the case that $v$ is saved.

	We return \yes if~$\DP[s_\Food,C_1,C_2,\emptyset,\kbar] \le \Dbar$, for some~$C_1, C_2\subseteq [N]$ which are disjoint.
	Otherwise, we return \no.
	
	\proofpara{Correctness}
	Let us quickly consider the basic case.
	For a leaf~$v$ in $T_\Food$, the set~$T_\Food^{(v)}$ only contains~$v$ and so~$\SSS_{(T_\Food^{(v)},C_1,C_2,\Phi,\ell)}$ for~$\ell>1$ is empty.
	If~$\ell=0$, the only possible triples $(A,\chi_1, \chi_2)$ in~$\SSS_{(T_\Food^{(v)},C_1,C_2,\Phi,\ell)}$ have $A \subseteq \{v\}$, i.e. we must let $v$ survive or go extinct.
	We can only let~$v$ survive if $\{v\}$ is~$\Phi$-\pgviable.
	Thus we can store $0$ if~$\gamma_\Sigma(\Phi) \ge 1$ and otherwise the minimal value of~$\w(F_{A,\chi_2})$ is just $\w(F_{v,\chi_2(x)})$, and so we store the minimum weight of a set of respecting edges~$F_{v,C}$ such that~$c(F_{v,C}) = C_1$ and $C \subseteq C_2$.
	
	To show the correctness of~\Recc{rec:D-i} and~\Recc{rec:D-t}, we assume that each previous table entries are correct and
	we prove that if~$\DP_i[v,C_1,C_2,\Phi,\Psi,\ell] = d$, respectively~$\DP[v,C_1,C_2,\Phi,\ell] = d$, then
	there is a triple~$(A,\chi_1,\chi_2)$
	with~$\w(F_{A,\chi_2}) = d$ from~$\SSS_{(Q_i,C_1,C_2,P_i,\ell)}$, respectively~$\SSS_{(T_\Food^{(v)},C_1,C_2,\Phi,\ell)}$.
	Afterward, we show that for each such triple, it holds that
	$\DP_i[v,C_1,C_2,\Phi,\Psi,\ell] \ge \w(F_{A,\chi_2})$, respectively~$\DP[v,C_1,C_2,\Phi,\ell] \ge \w(F_{A,\chi_2})$

	We first show the correctness of~\Recc{rec:D-i}.
	Assume that~$\DP_{i+1}[v,C_1,C_2,\Phi,\Psi,\ell] = d$.
	Then, by \Recc{rec:D-i}, there is a~$d^1\in [d]_0$ and~$C_1' \subseteq C_1, C_2' \subseteq C_2, \ell'\in [\ell]_0$ such that
	$\DP_i[v,C_1',C_2',\Phi,\Psi,\ell'] = d^1$ and
	$\DP[w_{i+1},C_1 \setminus C_1',(C_2 \setminus C_2') \cup C_1',(\Phi \cup \Psi) \cap \GW(w_{i+1}),\ell-\ell'] = d-d^1 =: d^2$.
	Consequently, there is $(A^1,\chi_1^1,\chi_2^1)$
	in~$\SSS_{(Q_i,C_1',C_2' \cup (C_1\setminus C_1'),P_i,\ell')}$ such that~$\w(F_{A^1,\chi_2^1}) = d^1$
	and there is~$(A^2,\chi_1^2,\chi_2^2) \in \SSS_{(T_\Food^{(w_{i+1})},C_1\setminus C_1',(C_2 \setminus C_2'),(\Phi \cup \Psi) \cap \GW(w_{i+1}),\ell-\ell')}$
	such that~$\w(F_{A^2,\chi_2^2}) = d^2$.
	As~$Q$ and~$T_\Food^{(w_{i+1})}$ are disjoint, so also~$A^1$ and~$A^2$ are disjoint.
	We therefore define a set~$A := A^1 \cup A^2$, 
	and mappings~$\chi_i$ with~$\chi_i(x) = \chi_i^j(x)$ for each taxon~$x\in A^j$, $i,j\in \{1,2\}$.
	Then~$\w(F_{A,\chi_2}) = \w(F_{A^1,\chi_2^1}) + \w(F_{A^2,\chi_2^2}) = d^1 + d^2 = d$.
	It remains to show that~$(A,\chi_1,\chi_2)$ is in~$\SSS_{(Q_{i+1},C_1,C_2,P_{i+1},\ell)}$.
	Most of the conditions follow because~$A^1$ and~$A^2$ are disjoint and the axioms hold for the individual sets.
	The difficult part is to show that $\chi_1(x_i)$ and $\chi_2(x_j)$ are disjoint for $i \leq j$, in the case that $x_i \in A^1$ and $x_j \in A^2$. For this,
	we observe that~$\chi_1^1(A^1) \subseteq C_1' \subseteq C_1$ and~$\chi_2^2(A^2) \subseteq C_2\setminus C_2' \subseteq C_2$, and these are disjoint.
	
	Now assume~$(A,\chi_1,\chi_2)$
	is in~$\SSS_{(Q_{i+1},C_1,C_2,P_{i+1},\ell)}$.
	Let $A^1 := A \cap Q_i$ and let $\chi_i^1$ be the restriction of $\chi_i$ to $A^1$, for $i \in \{1,2\}$. Similarly let $A^2 := A \setminus Q_i$, and let $\chi_i^2$ be the restriction of $\chi_i$ to $A^2$, for $i \in \{1,2\}$.
	It is straightforward to check that
	$(A^1,\chi_1^1,\chi_2^1)$ is in $\SSS_{(Q_i,C_1',C_2'\cup{(C_1\setminus C_1')},P_i,\ell')}$ and~$(A^2,\chi_1^2,\chi_2^2) \in \SSS_{(T_\Food^{(w_{i+1})},C_1\setminus C_1',(C_2 \setminus C_2'),(\Phi \cup \Psi) \cap \GW(w_{i+1}),\ell-\ell')}$,
	where $C_1' = \chi_1(A^1), C_2' = \chi_2(A^1)\setminus C_1'$, and $\ell = |A^1|$.   \todos{Do this more in detail for the journal.}

	We next show the correctness of~\Recc{rec:D-t}.
	Assume that~$\DP[v,C_1,C_2,\Phi,\ell] = d$.
	By \Recc{rec:D-t},
	we have~$\DP_t[v,C_1,C_2,\Phi,\predatorsE{v},\ell] = d$ 
	or~$\min_{C_2' \subseteq C_2}~\DP_t[v,C_1 \setminus c(F_{v,C_2'}),C_2 \setminus C_2',\Phi,\emptyset,\ell-1] + \w(F_{v,C_2'}) = d$.
	In the former case, there is~$(A,\chi_1,\chi_2)$ in~$\SSS_{(Q_{t},C_1,C_2,P_{t},\ell)}$\todos{Slack formality} with~$\w(F_{A,\chi_2}) = d$.
	We observe that~$(A,\chi_1,\chi_2)$ is also in in~$\SSS_{(T_\Food^{(v)},C_1,C_2,\Phi,\ell)}$ which is sufficient for this case.
	in the latter case, fix~$C_2'$ such that~$\DP_t[v,C_1 \setminus c(F_{v,C_2'}),(C_2 \setminus C_2') \cup c(F_{v,C_2'}),\Phi,\Psi=\emptyset,\ell-1] + \w(F_{v,C_2'}) = d$.
	Consequently, there is a triple~$(A,\chi_1,\chi_2)$ in~$\SSS_{(Q_{t},C_1 \setminus c(F_{v,C_2'}),(C_2 \setminus C_2') \cup c(F_{v,C_2'}),P_{t},\ell)}$.
	Consider~$(A',\chi_1\,\chi_2')$ with~$A' := A \cup \{v\}$, 
	$\chi_i'$ is~$\chi_i$ on~$A$ with~$\chi_2(v) = C_2'$ and~$\chi_1'(v) := c(F_{v,C_2'})$.
	Clearly~$\w(F_{A',\chi_2}) = \w(F_{A,\chi_2}) + \w(F_{v,C_2'}) = d$.
	As $(T_\Food^{(v)}\setminus A' = Q_t \setminus A$ and $Q_t \setminus A$ is $P_t$-\pgviable, we have also that $T_\Food^{(v)}\setminus A'$ is $\phi$-\pgviable (note that the only edges in $P_t \setminus \phi$ are those incoming at $v$, which are not needed for a set not containing $v$.)
	It remains to show that~$(A',\chi_1',\chi_2')$ is perfect, assuming that~$(A,\chi_1,\chi_2)$ is perfect.
	
	%
	Note that $v$ appears before any vertex in $Q_t$ in the depth-first traversal of $T_\Food$, and so we may assume $v$ is the first element of $A'$.
	It is clear that~$\chi_1'(v) = c(F_{v,C_2'})$ is disjoint from~$C_1 \setminus c(F_{v,C_2'})$,
	and so $\chi_1'(x)$ and $\chi_1'(y)$ are disjoint for all $x,y \in A'$.
	As $\chi_2(y) \subseteq C_2 \setminus C_2' \cup (C_1\setminus  c(F_{v,C_2'})$) for all $y \in A$ and $C_1, C_2'$ are disjoint, we have that $\chi_2'(v) = C_2'$ and $\chi_2'(y)$ are disjoint. Therefore  $\chi_2'(x)$ and $\chi_2'(y)$ are disjoint for all $x,y \in A'$.
	Similarly as $\chi_1'(v) = c(F_{v,C_2'})$ and $\chi_2(y) \subseteq C_2 \setminus C_2' \cup (C_1\setminus  c(F_{v,C_2'})$ for $y \in A$, we have that $\chi_1'(x_i)$ and $\chi_2'(x_j)$ are pairwise disjoint for all $x_i,x_j \in A'$ with $i < j$.
	The existence of~$F_{x,C_2'}$ follows from~$\DP_t' + \w(F_{v,C_2'}) = d$.

	On the converse, if~$(A,\chi_1,\chi_2)$ is in~$\SSS_{(T_\Food^{(v)},C_1,C_2,\Psi,\ell)}$, then we can show by a case distinction and with arguments similar to the previous paragraph that~$\DP[v,C_1,C_2,\Phi,\ell] \le \w(F_{A,\chi_2})$.

	\proofpara{Running time}
	By \Cref{lem:FxC-compute}, we can compute~$F_{x,C}$ or conclude that it does not exist for all~$x\in X$ and~$C\subseteq [N]$ in~$\Oh(2^N \cdot \sqrt{N} \cdot n \cdot |E(\Net)|)$ time.
	
	We observe that because~$C_1$ and~$C_2$ are disjoint, $T_\Food$ is a tree and therefore any vertex has at most one parent, and the field~$\Psi$ can only take two values for a fixed vertex~$v$.
	
	In the basic cases, in \Recc{rec:D-i}, and in \Recc{rec:D-t}, we iterate over~$C_2'\subseteq C_2$ and in \Recc{rec:D-i} we additionally iterate over~$C_1' \subseteq C_1$.
	Any color~$c\in [N]$ can therefore be in~$[N]\setminus (C_1 \cup C_2),C_2\setminus C_2',C_2',C_1\setminus C_1'$ or in the case of \Recc{rec:D-i} also in~$C_1'$.
	
	Thus, all table entries can be computed within~$\Oh(5^N 2^{\sw} \cdot n \cdot (N \cdot \kbar + |E(\Net)|))$ time.
\end{proof}
}

It remains to show how to reduce instances of \WMAPPDD to instances of \exWMAPPDD{N}.
This is done using standard color-coding techniques.

\newcommand{\ProofDbarCorr}{
	\proofpara{Correctness}
	For any subset of edges $E'$ with $|E'| \le N$, there is some~$f \in \mathcal{H}$ such that $c_f$ is injective on~$E'$ by~\Cref{def:perfectHashFamily}.
	
	Let~$\Instance_f$ be a \yes-instance of~\exWMAPPDD{N}.
	Thus, there is a perfect triple~$(A,\chi_1,\chi_2)$ such that~$X\setminus A$ is~\gviable,~$A$ has a size of at least~\kbar and for each~$x\in A$ there is a set of respecting edges~$F_{x,\chi_2(x)}\subseteq E(\Net)$ with~$\chi_1(x) = c(F_{x,\chi_2(x)})$ and~$\sum_{x\in A} \w(F_{x,\chi_2(x)}) \le \Dbar$.
	We show that~$S := X\setminus A$ is a solution for instance~$\Instance$.
	By definition,~$S$ is \gviable and has a size of~$|X| - |A| \le |X| - \kbar = k$.
	As~$(A,\chi_1,\chi_2)$ is perfect,~$F_{x,\chi_2(x)}$ have pairwise disjoint colors for~$x,y\in A$ and are therefore disjoint.
	We conclude with~\Cref{lem:perfectTripleBound} that~$\PD(S) \geq \w(E(\Net)) - \sum_{x\in A} \w(F_{x,\chi_2(x)}) \ge \w(E(\Net)) - \Dbar = D$.
	
	Now, let~$\Instance$ be a \yes-instance of~\WMAPPDD with solution~$S\subseteq X$.
	Let~$A := X\setminus S$ and since we can assume that~$S$ has a size of~$k$, the size of~$A$ is~$\kbar$.
	%
	Let~$E_A \subseteq E(\Net)$ be the set of edges~$uv$ where~$\off(v) \subseteq A$.
	As~$D \le \PD(S) = \w(E(\Net) \setminus E_A)$, we conclude that~$\w(E_A) \le \Dbar$.
	
	Now define~$F_{x_1}$ to be the edges~$uv$ in~$E_A$ with~$x_1 \in \off(v)$, and define~$F_{x_i}$ to be the edges~$uv$ in~$E_A \setminus \bigcup_{j=1}^{i-1} F_{x_j}$ with~$x_i \in \off(v)$.
	We define sets~$F_{x_i}'$ as follows.
	For each edge~$uv \in F_{x_i}$, for which no incoming edge of~$u$ is in~$F_{x_i}$,
	we add an edge~$uw$ outgoing of~$u$ to~$F_{x_i}'$, where~$uw$ is from~$Z_i := ((E(\Net) \setminus E_A) \cup \bigcup_{j=1}^{i-1} F_{x_i}) \setminus \bigcup_{j=1}^{i-1} F_{x_i}'$.
	If~$F_{x_i}$ contains more than one edge outgoing of~$u$, adding one edge is sufficient.
	\begin{claim}
		The set~$Z_i$ contains an edge outgoing of~$u$.
	\end{claim}
	\begin{claimproof}
		Assume first that~$\bigcup_{j=1}^{i-1} F_{x_i}'$ contains an edge~$uv'$.
		Without loss of generality, let~$uv' \in F_{x_t}'$ and~$\bigcup_{j=t+1}^{i-1} F_{x_i}'$ does not contain edges outgoing of~$u$.
		Then there is an edge~$uv_t \in F_{x_t}$ and by construction~$uv_t$ is not in~$F_{x_j}'$ for any~$j \in [i]$.
		Consequently, the set~$Z_i$ contains~$uv_t$.
		
		Now, assume that~$\bigcup_{j=1}^{i-1} F_{x_i}'$ does not contain edges outgoing of~$u$.
		If~$E(\Net) \setminus E_A$ contains an edge~$e$ outgoing of~$u$, then~$e \in Z_i$.
		Otherwise, all edges incoming at~$u$ are in~$E_A$.
		As~$F_{x_i}$ contains no edge incoming at~$u$, these edges and at least one edges~$e$ outgoing of~$u$ have to be in~$F_{x_t}$ for some~$t\in [i-1]$.
		We conclude~$e\in Z_i$ and the set~$Z_i$ contains an edge outgoing of~$u$.
	\end{claimproof}
	
	The sets~$F_{x_i}$ are constructed to be the respecting sets and~$F_{x_i}'$ are the auxiliary sets from the definition of respecting sets.
	We observe that~$E_A$ is the union of all~$F_{x_i}$.
	Thus, if~\Net is a tree, then~$|E_A| \le 2\kbar -1$, as a forest with~$\kbar$ leaves has at most~$2\kbar-1$ edges.
	If~\Net is a network, then~$|E_A| \le \w(E_A) = \w(E(\Net)) - \PD(S) \le \Dbar$.
	As we added at most one edge to~$F_{x_i}'$ per edge of~$F_{x_i}$, we conclude that the union~$U$ of all~$F_{x_i}$ and~$F_{x_i}'$ contains at most~$4\kbar -2$ if~\Net is a tree, and~$2\Dbar$ edges if~\Net is a network.

	Consequently, there is some $f \in \mathcal{H}$, such that $c_f$ is injective on~$U$.
	We set~$\chi_1(x_i) = c_f(F_{x_i})$ and~$\chi_2(x_i) = c_f(F_{x_i}')$.
	By the construction we conclude that~$(A,\chi_1,\chi_2)$ is perfect,~$S$ is~\gviable, and~$\sum_{i=1}^\kbar \w(F_{x_i,\chi_2(x_i)}) = \sum_{i=1}^\kbar \w(F_{x_i}) = \w(E_A) \le \kbar$.
	Thus,~$\Instance_f$ is a \yes-instance of~\exWMAPPDD{N}.

	\proofpara{Running Time}
	The construction of $\mathcal{H}$ takes $e^{N} N^{\Oh(\log {N})} \cdot q \log q$ time (\Cref{prop:perfectHashFamily}), and for each $f \in \mathcal{H}$ the construction of instance~$\Instance_f$ of \exWMAPPDD{N} takes time linear in $|\Instance|$.
	By~\Cref{lem:fkSMAPPD}, solving instances of \exWMAPPDD{N} takes $\Oh(5^N 2^{\sw} \cdot n \cdot (N^2 \cdot \kbar^2 + |E(\Net)|))$~time, and the number of instances is $|\mathcal{H}| = e^{N} N^{\Oh(\log {N})} \cdot \log q$.
	
	Thus, the total running time is 
	$\Oh(e^N N^{\Oh(\log N)} \log q \cdot (q + 5^N 2^{\sw} \cdot n \cdot (N^2 \cdot \kbar^2 + E(\Net))))$.
	This simplifies to $\Oh((5e)^N\cdot 2^{\sw + \Oh(\log^2(N))} \cdot n \cdot |E(\Net)| \cdot \log |E(\Net)|)$, as~$\kbar \le N$.
	
	Inserting~$2\Dbar$ or~$4\kbar-2$ into~$N$, respectively, gives the desired running times.
}
\newcommand{\ProofDbar}[1]{
\begin{proof}[Proof of~\Cref{thm:Dbar+sw}]
	\proofpara{Reduction}
	Let $\Instance = (\Net, \Food, k, D)$ be an instance of \PROB{Map-Weighted-PDD}.
	If~\Net is a tree, then set~$N$ to~$4\kbar-2$ and otherwise to~$2\Dbar$.
	
	Arbitrarily order the edges $e_1, \dots, e_{q}$ of $\Net$.
	We may assume~$q > N$, as otherwise, we can consider a single instance of \exWMAPPDD{q}.
	Let $\mathcal{H}$ be a~$(q, N)$-perfect hash family.
	For every $f \in \mathcal{H}$ we define a coloring $c_f$ by~$c_f(e_j) = f(j)$ for~$j\in [q]$
	and let~$\Instance_{f} = (\Net, \Food, k, D, c_f)$ be the corresponding instance of \exWMAPPDD{N}.
	Solve every instance $\Instance_f$, and return \yes if and only if $\Instance_{f}$ is a \yes-instance for some $f \in \mathcal{H}$.

	#1
\end{proof}
}
\ProofDbar{\emph{The proof of the correctness and running time is deferred to the appendix.}}

\section{Parameter~${k + \sw + \delta + h}$}
\label{sec:k}

By \Cref{thm:MAPPD}, \MAPPDD and \fMAPPDD are \Wh{2}-hard with respect to~$k+h$, even if the food web does not contain edges.
In the following, we therefore add the maximum in-degree of a reticulation~$\delta$ as a parameter and show that even the more general problem \WMAPPDD is \FPT with respect to~$k + \sw + h + \delta$.

To do this, we consider a parameter~$\Hgt$ generalizing the height of a tree.
\Hgt is the maximum number of tree edges that, in \Net, are on a path from the root to any taxon.
We observe that in a phylogenetic tree, $\Hgt$~is the height of the tree minus one.
Next, we prove bounds on the value of~$\Hgt$ in a phylogenetic network.

\newcommand{\lemSBar}[1]{
\begin{lemma}[$\star$]
	{#1}
	$h_t \le \Hgt$ and~$\Hgt \le \delta ^ {h_r} \cdot h_t \le \delta ^ h$.
\end{lemma}
}
\lemSBar{\label{lem:Sbar}}
\newcommand{\ProofSBar}{
\begin{proof}
	Recall $h_r$ (and~$h_t$) is the maximum number of reticulations (respectively tree vertices) on a path from the root to a leaf.
	For each path~$P$ from the root to a leaf we observe~$|E(P) \cap E_T(\Net)| \le \Hgt$ and as there is a path of containing~$h_t$ tree edges, we conclude~$h_t \le \Hgt$. 

	\looseness=-1
	To see that~$\Hgt \le \delta ^ {h_r} \cdot h_t \le \delta ^ h$, let $\Net$ be a network maximizing the value of~$\Hgt$ for values of $h_r, h_t, \delta$, and let~$x\in X$ be a leaf for which the maximum number of tree vertices with a path to~$x$ is maximized. We show that there is no tree vertex below a reticulation on any path to $x$.
	Assume for a contradiction that there is an edge~$rv_1 \in E(\Net)$ with~$v_1 \in V_T(\Net)$ and~$r \in V_R(\Net)$ and a path from~$v_1$ to~$x$.
	Let~$u_1,\dots,u_s$ be the parents of~$r$ and~$w$ a child of $v_1$ such that~$w$ has a path to~$x$.
	Remove the edges~$u_i r$ for~$i\in [s]$, $rv_1$, and~$v_1w$, add vertices~$v_2,\dots,v_{s}$ with attached leaves, and add edges~$u_iv_i$ for~$i\in [s]$, $v_i r$ for~$i\in [s]$ and~$rw$. Observe that after this transformation, the values of $h_r$, $h_t$, and $\delta$ have not changed, but there are~$s-1$ further tree vertices with a path to~$x$.
	As~$s>1$, this contradicts the maximality.
	
	We may therefore assume that no tree vertices with a path to $x$ are below a reticulation.
	We conclude that from the root at most~$\delta^{h_r}$ different paths of length~$h_t$ of tree vertices can lead to a leaf.
	\Cref{fig:worstCaseSbar} shows this scenario.
\end{proof}

\begin{figure}[t]
	\centering
	\begin{tikzpicture}[scale=.3]
		\node [textnode] at (-5, 3.5) {Length $h_t$};
		\node [textnode] at (-5, -1.5) {Length $h_r$};
		
		\node (0) at (0, 0) {};
		\node (1) at (8, 0) {};
		\node (2) at (4, -3) {};
		\node [leaf] (3) at (4, -4) {};
		\node [small,vertex] (12) at (4, 8) {};
		\node (13) at (-0.5, 7) {};
		\node (14) at (-0.5, 0) {};
		\node (15) at (-1.5, 3.5) {};
		\node (16) at (-0.5, 0) {};
		\node (17) at (-0.5, -3) {};
		\node (18) at (-1.5, -1.5) {};
		
		\draw (2.center) to (3);
		\draw (1.center) to (0.center);
		\draw [bend left=33] (0.center) to (2.center);
		\draw [bend left=33] (2.center) to (1.center);
		\foreach \i in {0,...,8} {
			\draw (\i,7) to (\i,0);
			\draw (\i,7) to (12);
		}
		\draw [in=-180, out=-15] (15.center) to (14.center);
		\draw [in=15, out=-180] (13.center) to (15.center);
		\draw [in=-180, out=-15] (18.center) to (17.center);
		\draw [in=15, out=-180] (16.center) to (18.center);
	\end{tikzpicture}
	\caption{\looseness=-1
		An example for \Cref{lem:Sbar}  where~$\Hgt$ is maximized.
		This is the case if above the leaf an upside-down pyramid of reticulations is followed by~$\delta^{h_r}$ paths of tree vertices of length~$h_t$.
	}
	\label{fig:worstCaseSbar}
\end{figure}
}

In this section we prove the following.

\begin{theorem}
	\label{thm:k+sw+S}
	Given a tree extension~$T_\Food$ of the food web with scanwidth~$\sw$,
	\PROB{Map-Weighted-PDD} can be solved in~$\Oh(2^{2.443\cdot k\Hgt + \sw + \Oh(\log^2(k\Hgt))} \cdot \sw \cdot n \cdot |E(\Net)|^2 \cdot \log |E_T(\Net)|)$ time.
\end{theorem}

As \MAPPDD and \fMAPPDD are special cases of \WMAPPDD, this result transfers to them.
\begin{corollary}
	\label{cor:k+sw+S-1PDD}
	\MAPPDD and \fMAPPDD are \FPT with respect to~$k+\sw+h+\delta$%
	\ifConference
	.
	\fi
\end{corollary}

First we define some objects necessary in this chapter.
%
For a a set~$A$ of taxa and a mapping~$\chi: X \to 2^{[k\cdot \Hgt]}$, we say
that tuple~$(A,\chi)$ is \emph{colorful} if $\chi(x) \cap \chi(y) \ne \emptyset$ for~$x,y\in A$ implies $x=y$.

\begin{definition}
Fix a mapping~$c: E_T(\Net) \to [k\cdot \Hgt]$. For a taxon~$x\in X$ and a mapping~$\chi: X \to 2^{[k\cdot \Hgt]}$, 
A set of edges~$F\subseteq E(\Net)$ is \emph{$(x,\chi)$-suitable}, if
\begin{itemize}
	\item $c(e) \in \chi(x)$ for each edge~$e \in F \cap E_T(\Net)$,
	\item $c(e_1) \ne c(e_2)$ for each pair~$e_1,e_2 \in F \cap E_T(\Net)$,
	and
	\item for each edge~$uv \in F$,
	there is a path from~$v$ to~$x$ in~$\Net[F]$.
\end{itemize}
When $(x,\chi)$ is clear from the context we say an $(x,\chi)$-respecting set is \emph{respecting}.
\end{definition}

\Cref{fig:exampleSuitable} shows an example.
Note that a respecting set may contain reticulation edges.

\begin{figure}
	\DeclareRobustCommand{\tikzdot}[1]{\tikz[baseline=-0.6ex]{\node[draw,fill=#1,inner sep=2pt,circle] at (0,0) {};}}
	\definecolor{col1}{HTML}{e6ab02}
	\definecolor{col2}{HTML}{e7298a}
	\definecolor{col3}{HTML}{1b9e77}
	
	\definecolor{bgcol1}{HTML}{ebebeb}
	\definecolor{bgcol2}{HTML}{c8e6c9}
	\definecolor{bgcol3}{HTML}{bbdefb}
	
	\centering
	\begin{tikzpicture}
		\node[textnode] at (5,-1.5) {
			\begin{tabular}{c|c}
				& $\chi(x)$ \\
				\hline
				$x_1$ & $\{\tikzdot{col1}\}$ \\
				$x_2$ & $\{\tikzdot{col2}\}$ \\
				$x_3$ & $\{\tikzdot{col2},\tikzdot{col3}\}$ \\
				$x_4$ & $\{\tikzdot{col1},\tikzdot{col2}\}$ \\
				$x_5$ & $\emptyset$
			\end{tabular}
		};
		
		\coordinate (bg0) at (0, 0) {};
		\coordinate (bg1) at (-1.5, -0.75) {};
		\coordinate (bg2) at (0, -0.75) {};
		\coordinate (bg3) at (1.5, -0.75) {};
		\node [bgvertex, fill=bgcol3] (bg4) at (-0.75, -1.5) {};
		\node [bgvertex, fill=bgcol1] (bg5) at (0.75, -1.5) {};
		\coordinate (bg6) at (0.75, -2.25) {};
		\node [bgvertex, fill=bgcol2] (bg11) at (-2, -1.5) {};
		\node [bgvertex, fill=bgcol3] (bg12) at (-0.75, -2.25) {};
		\node [bgvertex, fill=bgcol3] (bg13) at (0.25, -3) {};
		\node [bgvertex, fill=bgcol1] (bg14) at (1.25, -3) {};
		\node [bgvertex, fill=bgcol2] (bg15) at (2, -1.5) {};
		
		\draw [line width=12pt,color=bgcol2] (bg0.center) to (bg1.center);
		\draw [line width=12pt,color=bgcol3] (bg0.center) to (bg2.center);
		\draw [line width=7pt,color=bgcol1] (bg0.center) to (bg2.center);
		\draw [line width=12pt,color=bgcol3] (bg1.center) to (bg4.center);
		\draw [line width=12pt,color=bgcol3] (bg2.center) to (bg4.center);
		\draw [line width=12pt,color=bgcol1] (bg2.center) to (bg5.center);
		\draw [line width=12pt,color=bgcol1] (bg3.center) to (bg5.center);
		\draw [line width=12pt,color=bgcol1] (bg5.center) to (bg6.center);
		\draw [line width=12pt,color=bgcol3] (bg4.center) to (bg12.center);
		\draw [line width=12pt,color=bgcol3] (bg6.center) to (bg13.center);
		\draw [line width=12pt,color=bgcol2] (bg1.center) to (bg11.center);
		\draw [line width=12pt,color=bgcol2] (bg3.center) to (bg15.center);
		\draw [line width=12pt,color=bgcol1] (bg6.center) to (bg14.center);
		
		\newcommand{\mcvertex}[7]{
			\fill [color=bgcol#2] (#1) -- +(60:6pt) arc [start angle=60, end angle=0, radius=6pt] -- cycle;
			\fill [color=bgcol#3] (#1) -- +(0:6pt) arc [start angle=0, end angle=-60, radius=6pt] -- cycle;
			\fill [color=bgcol#4] (#1) -- +(-60:6pt) arc [start angle=-60, end angle=-120, radius=6pt] -- cycle;
			\fill [color=bgcol#5] (#1) -- +(-120:6pt) arc [start angle=-120, end angle=-180, radius=6pt] -- cycle;
			\fill [color=bgcol#6] (#1) -- +(180:6pt) arc [start angle=180, end angle=120, radius=6pt] -- cycle;
			\fill [color=bgcol#7] (#1) -- +(120:6pt) arc [start angle=120, end angle=60, radius=6pt] -- cycle;
		}
		\mcvertex{bg0}321321
		\mcvertex{bg0}321321
		\mcvertex{bg1}323232
		\mcvertex{bg2}131313
		\mcvertex{bg3}212121
		\mcvertex{bg6}131313

		\node [vertex, small] (0) at (bg0) {};
		\node [vertex, small] (1) at (bg1) {};
		\node [vertex, small] (2) at (bg2) {};
		\node [vertex, small] (3) at (bg3) {};
		\node [reti, small] (4) at (bg4) {};
		\node [reti, small] (5) at (bg5) {};
		\node [vertex, small] (6) at (bg6) {};
		\node [leaf, label={left:$x_1$}] (11) at (bg11) {};
		\node [leaf, label={left:$x_2$}] (12) at (bg12) {};
		\node [leaf, label={left:$x_3$}] (13) at (bg13) {};
		\node [leaf, label={right:$x_4$}] (14) at (bg14) {};
		\node [leaf, label={right:$x_5$}] (15) at (bg15) {};
		
		\draw [color=col1, ultra thick] (0) to (1);
		\draw [color=col2, ultra thick] (0) to (2);
		\draw [color=col3, ultra thick] (0) to (3);
		\draw (1) to (4);
		\draw (2) to (4);
		\draw (2) to (5);
		\draw (3) to (5);
		\draw [color=col1, ultra thick] (5) to (6);
		\draw (4) to (12);
		\draw (6) to (13);
		\draw (1) to (11);
		\draw (3) to (15);
		\draw (6) to (14);
	\end{tikzpicture}
	\caption{\looseness=-1
	A hypothetical network on~3 colors~$\{\tikzdot{col1},\tikzdot{col2},\tikzdot{col3}\}$.
	For each taxon the $(x,\chi)$-suitable set~$F_{x,\chi}$ is indicated in~$\tikzdot{bgcol1}$,~$\tikzdot{bgcol2}$, or~$\tikzdot{bgcol3}$.
	In this network,~$\Hgt$ takes the value of~3---there are three tree edges on paths to~$x_3$ or~$x_4$.
	}
	\label{fig:exampleSuitable}
\end{figure}

To prove~\Cref{thm:k+sw+S}, we first show how to solve a colored variant of~\PROB{Map-Weighted-PDD}.
In~\kSWMAPPDD, besides the usual input
\ifConference
$(\Net,\Food,k,D)$%
\fi
, we are given a mapping~$c: E_T(\Net) \to [k\cdot \Hgt]$ of a color per tree edge.
We ask whether a~\gviable set of taxa~$S\subseteq X$ of size at most~$k$ and a mapping~$\chi: S \to 2^{[k\cdot \Hgt]}$ exist such that~$(S,\chi)$ is colorful and for each~$x\in S$ there is a set of suitable edges~$F_{x,\chi}$ and~$\sum_{x\in S} \w(F_{x,\chi}) \ge D$.

\begin{lemma}
	\label{lem:optimal-F}
	Given an instance~$\Instance$ of~\kSWMAPPDD, a leaf~$x$, and a mapping~$\chi$, a suitable edge set~$F$ with maximal value~$\w(F)$ can be computed in~$\Oh(2^H \cdot |E(\Net)| \cdot (kH + |E(\Net)|))$~time. 
\end{lemma}
Because of this lemma, in the rest of the section, we simply will write~$\w(F_{x,\chi})$ when referring to the weight of a suitable edge set for~$x$, which has maximal weight.
\begin{proof}
	\proofpara{Algorithm}
	Let~$E_x$ be the set of tree edges~$uv$ with~$x\in \off(v)$.
	For a taxon~$x$ and a mapping~$\chi$, iterate over all subsets~$F$ of~$E_x$.
	If~$uv\in F$ and~$u$ is a reticulation, then add all incoming edges of~$u$ to~$F$, for each~$uv\in F$.
	Add the edge incoming at~$x$ to~$F$.
	Return the biggest value of a so computed~$F$ which is suitable.
	
	\proofpara{Correctness}
	If the algorithm returns value~$d$, then there is a suitable set~$F$ with~$d=\w(F)$.
	
	Let~$F$ be a suitable set.
	The set~$F_x := F\cap E_x$ appears in the iteration.
	The edge incoming at~$x$ is in~$F$, as it is suitable.
	Reticulation edges have no color.
	Thus, adding reticulation edges leading to~$F$ keeps the set suitable.
	Consequently, a set~$F'$ with~$F'\supseteq F$ has been considered by the algorithm.
	
	\proofpara{Running time}
	The set~$E_x$ has size at most~$H$, by definition.
	All sets~$F$ are computed in~$\Oh(2^H \cdot |E(\Net)|)$ time.
	Checking whether a set~$F$ is colorful can be done in~$\Oh(kH + |E(\Net)|)$~time.
\end{proof}

\begin{lemma}
	\label{lem:disjoint-Fx}
	If~$(S,\chi)$ is colorful, then any $(x,\chi)$-suitable set~$F_{x,\chi}$ and any $(y,\chi)$-suitable set~$F_{y,\chi}$ are pairwise disjoint for~$x, y\in S$, $x\ne y$.
\end{lemma}
\begin{proof}
	Fix vertices~$x,y\in S$ and an edge~$uv\in F_{x,\chi}$.
	If~$uv$ is a tree edge, then~$c(uv) \in \chi(x)$ and because~$\chi(x)$ and~$\chi(y)$ are disjoint, we conclude~$uv \not\in F_{y,\chi}$.
	
	Now let~$uv$ be a reticulation edge.
	Let~$w$ be the unique \emph{first} tree vertex on a path from~$v$ to~$x$---or any other leaf.
	Then~$c(e_w) \in \chi(x)$ for the edge~$e_w$ incoming at~$w$.
	We conclude that~$e_w \not\in F_{y,\chi}$ and there is no path from~$v$ to~$y$ in~$\Net[F_{y,\chi}]$.
	Consequently~$uv \not\in F_{y,\chi}$.
	
	If~$v=x$, or the only taxon reachable from~$v$ is~$x$, then~$uv \not\in F_{y,\chi}$.
\end{proof}

We write~$F_{S,\chi}$ for the union~$\bigcup_{x\in S} F_{x,\chi}$ for a set~$S$.
For a colorful set~$S$ we conclude with~\Cref{lem:disjoint-Fx} that~$\PD(S) \ge \w(F_{S,\chi}) \ge D$.

We are now ready to prove how to solve instances of~\kSWMAPPDD.

\newcommand{\lemfkSMAPPD}[1]{
\begin{lemma}[$\star$]
	{#1}
	Given an instance~$\Instance$ of~\kSWMAPPDD and a tree-ex\-ten\-sion~$T_\Food$ of the food web, we can solve~$\Instance$ in~$\Oh(2^{k\cdot \Hgt + \sw} \cdot k^4 \cdot \Hgt^2 \cdot {\sw} \cdot n \cdot |E(\Net)|^2)$~time.
\end{lemma}
}
\lemfkSMAPPD{\label{lem:fkSMAPPD}}
\newcommand{\IntuitionfkSMAPPD}{
	We consider the tree extension of the food web bottom up in a dynamic programming algorithm. We track the existence of colorful tuples whose taxa are all below a given vertex of the tree extension. We index colorful tuples by the sets of colors used, as well as by the set of edges $\Phi$ for which the set of taxa is $\Phi$-\pgviable, and we store the total weight of the suitable edge sets.
	We use the fact that the sets~$\chi(x)$ and~$\chi(y)$ have to be pairwise disjoint.
	Therefore, if a taxon~$x$ is saved and a set of colors has been assigned, these colors can be removed.
	We define~$\chi(x)$ when selecting~$x$.
}
The intuition behind~\Cref{lem:fkSMAPPD} is as follows:
\IntuitionfkSMAPPD
\newcommand{\ProoffkSMAPPD}{
\begin{proof}
	\proofpara{Intuition}
	\IntuitionfkSMAPPD

	\proofpara{Table Definition}
	For a set of taxa~$Q\subseteq X$, a set of colors~$C\subseteq [k \cdot \Hgt]$, a set of edges~$\Phi \subseteq E(\Food)$ between~$Q$ and~$X\setminus Q$, and an integer~$\ell\in [k]_0$, let~$\SSS_{(Q,C,\Phi,\ell)}$ be the set of colorful tuples~$(S,\chi)$ of a set~$S \subseteq Q$ of size \emph{at most}~$\ell$ and a mapping~$\chi: S \to 2^C$ with~$\chi(S) = C$ such that~$S$ is~$\Phi$-\pgviable.
	We define a dynamic programming algorithm with table~$\DP$.
	For a vertex~$v\in V(T_\Food)$, a set of colors~$C\subseteq [k \cdot \Hgt]$, a set of edges~$\Phi \subseteq \GW(v)$, and an integer~$\ell\in [k]_0$, we store in~$\DP[v,C,\Phi,\ell]$ the \emph{maximum} value~$\w(F_{S,\chi})$ of any tuple~$(S,\chi) \in \SSS_{(T_\Food^{(v)},C,\Phi,\ell)}$.
	The value of an optimal set can then be found in~$\DP[s_\Food,[k\cdot \Hgt],\emptyset,k]$.
	Is~$v$ an internal vertex of~$T_\Food$ with children~$w_1,\dots,w_t$,
	then we define further auxiliary tables~$\DP_i[v,C,\Phi,\Psi,\ell]$ analogously, where tuples~$(S,\chi)$ are considered from~$\SSS_{\left(Q,C,P,\ell\right)}$ with~$Q:=\bigcup_{j=1}^i T_\Food^{(w_j)}$ and~$P:=(\Phi \cup \Psi) \cap \left(\bigcup_{j=1}^i \GW(w_j)\right)$;
	that is---only the first~$i$ children of~$v$ are considered.
	We require that the set of edges~$\Psi$ is either empty or~$\predatorsE{v}$.

	\proofpara{Algorithm}
	We define the table in a bottom-up fashion.
	Let~$v$ be a leaf of~$T_\Food$.
	(That is a taxon without predators).
	In~$\DP[v,C,\Phi,\ell]$, we store~$\w(F_{v,\chi})$ with~$\chi(v) = C$ if~$\gamma_\Sigma(\Phi) = 1$ and~$\ell \ge 1$.
	Otherwise, we store~$0$.
	
	Now let~$v$ be an internal vertex of~$T_\Food$ with children~$w_1,\dots,w_t$.
	We define~$\DP_1[v,C,\Phi,\Psi,\ell]$ to be~$\DP[w_1,C,(\Phi \cup \Psi) \cap \GW(w_1),\ell]$.
	To compute further values for~$i\in [t-1]$, we use the following recurrence in which we define~$\DP_{i+1}[v,C,\Phi,\Psi,\ell]$ to be
	\begin{eqnarray}
		\label{rec:fkS-i}
			\max_{C'\subseteq C, \ell'\in [\ell]_0}
			\DP_i[v,C',\Phi,\Psi,\ell'] + 
			\DP[w_{i+1},C\setminus C',(\Phi \cup \Psi) \cap \GW(w_{i+1}),\ell-\ell'].
	\end{eqnarray}
	Finally, if~$\ell \ge 1$ and~($v=s_\Food$ or~$\gamma_\Sigma(\Phi\cap\preyE{v}) \ge 1$), we set~$\DP[v,C,\Phi,\ell]$ to be
	\begin{eqnarray}
		\label{rec:fkS-t}
		\max\{~\DP_t[v,C,\Phi,\emptyset,\ell]~;~\max_{C' \subseteq C}~\DP_t[v,C \setminus C',\Phi,\predatorsE v,\ell-1] + \w(F_{v,\chi})~\}.
	\end{eqnarray}
	Here, we use~$\chi(v) := C'$.
	Otherwise, we set $\DP[v,C,\Phi,\ell]$ to be $\DP_t[v,C,\Phi,\emptyset,\ell]$.
	
	We return \yes if~$\DP[s_\Food,[k\cdot \Hgt],\emptyset,k] \ge D$.
	Otherwise, we return \no.
	
	\proofpara{Correctness}
	We only show the correctness of~\Recc{rec:fkS-t} and omit to show the similar case in~\Recc{rec:fkS-i} as well as the basic cases.
	Assume that~$\DP$ stores the correct value for all children of~$v$ in~$T_\Food$ and~$\DP_t$ stores the correct value for~$v$.
	We show first that if~$\DP[v,C,\Phi,\ell] = d$, then there is a tuple~$(S,\chi) \in \SSS_{(T_\Food^{(v)},C,\Phi,\ell)}$ and~$\w(F_{S,\chi}) = d$.
	Afterward, we show that~$\DP[v,C,\Phi,\ell] \ge \w(F_{S,\chi})$ for each tuple~$(S,\chi) \in \SSS_{(T_\Food^{(v)},C,\Phi,\ell)}$.
	
	So assume that~$\DP[v,C,\Phi,\ell] = d$.
	Consequently, we have~$\DP_t[v,C,\Phi,\emptyset,\ell] = d$, or there is some~$C'\subseteq C$ such that~$\DP_t[v,C\setminus C',\Phi,\predatorsE{v},\ell-1] = d - \w(F_{v,\chi})$ for~$\chi(v) = C'$.
	In the former case, there is a colorful tuple~$(S,\chi)$ with a set~$S \subseteq \bigcup_{i=1}^t T_\Food^{(w_i)} = T_\Food^{(v)}\setminus \{v\}$ of size~$\ell$ and a mapping~$\chi$ such that~$\w(F_{S,\chi}) = d$ and~$S$ is~$(\Phi \cup \emptyset)$-\pgviable.
	Consequently,~$S$ is also~$\Phi$-\pgviable in~$T_\Food^{(v)}$.
	In the latter case, $\ell \ge 1$ and there is a colorful tuple~$(S,\chi)$ with a set~$S \subseteq T_\Food^{(v)}\setminus \{v\}$ of size~$\ell-1$ and a mapping~$\chi$ with~$\chi(S) = C\setminus C'$, such that~$\chi(v) = C'$ and~$\w(F_{S,\chi}) = d - \w(F_{v,\chi})$.
	Thus, adding~$v$ to~$S$ yields the desired set.
	
	Conversely, let~$(S,\chi)\in \SSS_{(T_\Food^{(v)},C,\Phi,\ell)}$ be given.
	Assume first that~$v$ is not in~$S$.
	Because~$S$ is $\Phi$-\pgviable, we conclude~$\DP[v,C,\Phi,\ell] = \DP_t[v,C,\Phi,\emptyset,\ell] \ge \w(F_{S,\chi})$.
	Now, let~$v$ be in~$S$.
	As~$S$ is $\Phi$-\pgviable and~$\preyE{x} \cap \GW(v) = \preyE{x} \cap \Phi$ for each~$x\in S$, the set~$S\setminus \{v\}$ is $(\Phi \cup \predatorsE{v})$-\pgviable.
	We conclude~$(S\setminus \{v\},\chi) \in \SSS_{(T_\Food^{(v)}\setminus \{v\},C,\Phi,\ell-1)}$ and therefore~$\DP_t[v,C,\Phi,\predatorsE{v},\ell-1] \ge \w(F_{S,\chi}) - \w(F_{v,\chi})$ and further~$\DP[v,C,\Phi,h] \ge \w(F_{S,\chi})$.

	\proofpara{Running time}
	Observe that for any vertex~$v$, only two values are possible for~$\Psi$.
	Therefore, all tables have~$\Oh(2^{k\cdot \Hgt + \sw} \cdot nk)$ entries, together.
	
	By~\Recc{rec:fkS-i}, in time~$\Oh(2^{k\cdot \Hgt} \cdot 2^{\sw} \cdot k^2 \cdot {\sw} \cdot (k\Hgt) \cdot n)$ all entries of $\DP_{i+1}$ can be computed, using convolutions.
	By~\Recc{rec:fkS-t}, we can compute~$\DP[v,C,\Phi,\ell]$ in~$\Oh(k\cdot \Hgt)$~time and~$\w(F_{v,\chi})$ needs to be computed once per vertex, which by~\Cref{lem:Sbar} can be done in~$\Oh(2^\Hgt \cdot |E(\Net)| \cdot (k\Hgt + |E(\Net)|))$ time.
	
	This leads to an overall running time of~$\Oh(2^{k\cdot \Hgt + \sw} \cdot k^4 \cdot \Hgt^2 \cdot {\sw} \cdot n \cdot |E(\Net)|^2)$.
\end{proof}
}

It remains to show how to reduce instances of \WMAPPDD to an instances of~\kSWMAPPDD.
For this, we use perfect hash families.
(\Cref{def:perfectHashFamily})

\begin{proof}[Proof of \Cref{thm:k+sw+S}]
	\proofpara{Reduction}
	Let $\Instance = (\Net, \Food, k, D)$ be an instance of \WMAPPDD.
	
	Arbitrarily order the tree edges $e_1, \dots, e_{q}$ of $\Net$.
	We may assume~$q > k\cdot \Hgt$.
	Let $\mathcal{H}$ be a~$(q, k\cdot \Hgt)$-perfect hash family.
	For every $f \in \mathcal{H}$ we define a coloring $c_f$ by~$c_f(e_j) = f(j)$ for~$j\in [q]$
	and let~$\Instance_{f} = (\Net, \Food, k, D, c_f)$ be the corresponding instance of \kSWMAPPDD.
	Now, solve instance $\Instance_f$, and return \yes if and only if $\Instance_{f}$ is a \yes-instance for some $f \in \mathcal{H}$.
	
	\proofpara{Correctness}
	For any set~$E' \subseteq E_T(\Net)$ of edges with a size of at most $k\cdot \Hgt$, there is a function~$f \in \mathcal{H}$ such that $c_f(E')$ contains each color at most once.

	Now, let~$\Instance$ be a \yes-instance of~\WMAPPDD with solution~$S=\{x_1,\dots,x_k\}\subseteq X$.
	Further, let~$E_T^{(1)}$ be the set of tree edges on paths from the root to~$x_1$ in \Net, and for~$i\in [k-1]$ let~$E_T^{(i+1)}$ be the set of tree edges on paths to~$x_{i+1}$ which are not in~$E_T^{(i)}$.
	We define~$E_T(S)$ as the union of these sets.
	By definition of~$\Hgt$, each set~$E_T^{(i)}$ has a size of at most~$\Hgt$.
	By definition of perfect hash families, there is some $f \in \mathcal{H}$, such that $c_f$ is injective on~$E_T(S)$.
	Taking~$\chi(x_i) = c_f(E_T^{(i)})$, we conclude that~$(S,\chi)$ is colorful and~$\w(F_{S,\chi}) = \PD(S) \ge D$.
	Thus,~$(S,\chi)$ is a solution of the \yes-instance~$\Instance_f$ of~\kSWMAPPDD.
	
	Conversely, whenever $(S,\chi)$ is a solution for instance~$\Instance_f$, then~$S$ is also a solution for $\Instance$.

	\proofpara{Running Time}
	The construction of $\mathcal{H}$ takes $e^{k\Hgt} ({k\Hgt})^{\Oh(\log {k\Hgt})} \cdot q \log q$ time, and for each $f \in \mathcal{H}$ the construction of instance~$\Instance_f$ of \kSWMAPPDD takes time linear in~$|\Instance|$.
	By~\Cref{lem:fkSMAPPD}, instances of \kSWMAPPDD can be solved in~$\Oh(2^{k\cdot \Hgt + \sw} \cdot k^4 \cdot \Hgt^2 \cdot {\sw} \cdot n \cdot |E(\Net)|^2)$~time, and the number of instances is $|\mathcal{H}| = e^{k\Hgt} ({k\Hgt})^{\Oh(\log {k\Hgt})} \cdot \log q$.
	
	Thus, the total running time is 
	$\Oh(e^{k\Hgt} ({k\Hgt})^{\Oh(\log {k\Hgt})} \log q \cdot (q + 2^{k\cdot \Hgt + \sw} \cdot k^4 \cdot \Hgt^2 \cdot \sw \cdot n \cdot |E(\Net)|^2))$.
	This simplifies to $\Oh((2e)^{k\Hgt}\cdot 2^{\sw + \Oh(\log^2(k\Hgt))} \cdot \sw \cdot n \cdot |E(\Net)|^2 \cdot \log |E_T(\Net)|)$.
\end{proof}

By \Cref{thm:PDD+}\ref{itm:k-fPDD}, \WMAPPDD is not \FPT with respect to~$k+h$.
But this differs when we add~$\sw$ as a parameter.
In a phylogenetic tree, the value of~$\Hgt$ is~$h$.

\begin{corollary}
	\label{prop:k+sw+h}
	\WPDD is \FPT with respect to~$k+h+\sw$ and can be solved\linebreak in~$\Oh(2^{2.443\cdot kh + \sw + \Oh(\log^2(kh))} \cdot \sw \cdot n \cdot |E(\Net)|^2 \cdot \log |E_T(\Net)|)$~time.
\end{corollary}

\subsection{\textbf{Map-${\varepsilon}$-PDD} with respect to~${k+\delta+h}$}
By \Cref{thm:k+sw+S}, \WMAPPDD is \FPT with respect to~$k+\sw+\delta+h$ and therefore also~\MAPPDD and~\fMAPPDD.
If any of the four parameters is dropped then \fMAPPDD is \Wh{1}-hard, as pointed out in~\Cref{tab:drop-1-para}.

\begin{table}
	\centering
	\caption{\fMAPPDD is \Wh{1}-hard if one of the parameters of~$k+\sw+\delta+h$ is dropped.}
	\begin{tabular}{cccc}
		\toprule
		$\sw+\delta+h$ & $k+\sw+\delta$ & $k+\sw+h$ & $k+\delta+h$ \\
		\Cref{thm:faller},\cite{faller} & \Cref{cor:binary} & \Cref{thm:MAPPD},\cite{bordewichNetworks,MAPPD} & \Cref{thm:PDD+}\ref{itm:k-fPDD},\cite{WeightedFW} \\
		\bottomrule
	\end{tabular}
	\label{tab:drop-1-para}
\end{table}

While the first three hardness results hold for both problems, the last explicitly only shows hardness for \fMAPPDD.
We expect this hardness not to hold for \MAPPDD, but believe that with an approach similar to the one presented in~\cite{PDD} to show that~\PDD is \FPT with respect to $k+h$, one can also show that \MAPPDD is \FPT when parameterized with $k+\delta+h$.
Unfortunately, the proof given in~\cite{PDD} has an incorrect lemma.
An example of the error is explain in greater detail in Appendix~\Cref{sec:counter-example}.
We still believe both of the claims to hold.
Henceforth, we leave the following as a conjecture.

\begin{conjecture}
	\label{con:k+d+h}
	\MAPPDD is \FPT when parameterized with $k+\delta+h$.
\end{conjecture}

\section{Discussion}
\label{sec:discussion}
In this paper, we made the approach to combine the natural problems of maximizing phylogenetic diversity in a network with viability constraints given through a food web.
We defined the problem~\WMAPPDD and its special cases \MAPPDD and~\fMAPPDD.
We provided several \FPT algorithms for these problems and even presented a complete complexity dichotomy by showing for which combination of parameters of~$k$, $\Dbar$, $\sw$, $\delta$, and~$h$, the three problems are in \FPT and for which they are~\Wh{1}-hard.

Still, several questions remain open.
The most obvious is whether \Cref{con:k+d+h} holds.

Further, in \Cref{sec:Dbar}, we showed that \WMAPPDD is \FPT with respect to~$\Dbar + \sw$.
We showed that this approach is sufficient to show that \WPDD is \FPT with respect to the smaller parameter~$\kbar + \sw$.
Yet, it is unclear whether \WMAPPDD admits an \FPT algorithm for this parameter.

Another major open question, already pointed out in~\cite{PDD}, is whether \PDD and \fPDD are \FPT when parameterized with $k$.
This question even remains open if the food web is restricted to an out-tree---or if any vertex in~\Food has a degree of at most~1.

In this paper, we observed the \emph{All-Paths-PD} measure in phylogenetic networks as a measure of phylogenetic diversity.
This measure is computationally the least challenging, but probably not the best measure in capturing phylogenetic diversity.
In recent years, several other measures have been proposed~\cite{bordewichNetworks,van2025average,van2024maximizing,WickeFischer2018}.
We wonder if there are also efficient algorithms for these problems that respect the viability of the selected set of taxa.

\thispagestyle{empty}
\ifArXiv

\else
\bibliography{ref}
\fi
\bibliographystyle{abbrv}
\thispagestyle{empty}

\ifConference
\setcounter{page}{1}
\setcounter{section}{0}
\renewcommand\thesection{A.\arabic{section}}

\renewcommand\thesection{A}
\section{Appendix---Omitted proofs}
\renewcommand\thesection{A.\arabic{section}}
\setcounter{section}{0}

\section{Proof of~\Cref{thm:Dbar+sw}}
\renewcommand\thetheorem{\ref{thm:Dbar+sw}}

\thmDbar{}
\ProofDbar{\ProofDbarCorr}

\section{Proof of~\Cref{lem:perfectTripleBound}}
\renewcommand\thelemma{\ref{lem:perfectTripleBound}}

\lemPerfectTripleBound{}
\ProofPerfectTripleBound{}

\section{Proof of~\Cref{lem:exPPD}}
\renewcommand\thelemma{\ref{lem:exPPD}}

\lemexPPD{}
\ProofexPPD

\section{Proof of~\Cref{lem:Sbar}}
\renewcommand\thelemma{\ref{lem:Sbar}}

\lemSBar{}
\ProofSBar

\section{Proof of~\Cref{lem:fkSMAPPD}}
\renewcommand\thelemma{\ref{lem:fkSMAPPD}}

\lemfkSMAPPD{}
\ProoffkSMAPPD

\renewcommand\thesection{B}
\section{Appendix---Counter-example for a theorem in~\cite{PDD}}
\label{sec:counter-example}
In this section, will briefly point out a mistake in Reduction Rule~6 of~\cite{PDD}.

\subsection{The Problem Definition}
Reduction Rule~6 of~\cite{PDD} is defined on a colored version of \PDD, which is used as a subroutine to prove that \PDD is \FPT with respect to $k+h$.
In this colored version of \PDD, additionally to a vertex-colored phylogenetic tree~\Tree, a food web~\Food, $k$, and~$D$, one is given a \emph{pattern tree}~$\Tree_P$---a colorful tree on the same set of colors as the phylogenetic tree has.
The task is to find an \viable set~$S$ of taxa with~$|S|\le k$, $\PDsub{\Tree}(S) \ge D$, and
the spanning tree of the root and~$S$ should be isomorphic with the pattern tree.

Reduction Rule~6 states as follows:
\renewcommand\therr{6}
\begin{rr}
	\label{rr:internal-vertex}
	Apply previous Reduction Rules exhaustively.
	Let~$\rho$ be the root of~\Tree and let~$\rho_P$ be the root of~$\Tree_P$.
	Let $v'$ be a grand-child of $\rho_P$ and let~$u'$ be the parent of~$v'$.
	\begin{enumerate}[1.]
			\item For each vertex $u$ of $\Tree$ with~$c(u) = c_P(u')$
			add edges $\rho v$ to $\Tree$ for every child $v$ of~$u$.
			\item Set the weight of $\rho v$ to be $\w(uv)$ if~$c(v) \ne c_P(v')$
			or~$\w(uv)+\w(\rho u)$ if~$c(v) = c_P(v')$.
			\item Add edges~$\rho_P w'$ to $\Tree_P$ for every child~$w'$ of~$u'$.
			\item Set~$\Tree_P' := \Tree_P - u'$ and~$\Tree' := \Tree - u$.
	\end{enumerate}
\end{rr}

Reduction Rule~6 is used to reduced the height of the phylogenetic tree until it is only a star and then to use a known algorithm for that case.

\subsection{The Error}
\Cref{fig:mistake} shows an example execution of Reduction Rule~6.
For~$D=41$ and~$k=2$ the original instance is a \no-instance, while in the latter, it is a \yes-instance with the two outermost vertices being a solution (without considering the food web.)

Unfortunately, it is not enough to drop the condition that the spanning tree of the root and~$S$ should be isomorphic with the pattern tree.
For this example consider~$D=42$ and~$k=2$ for which still the first instance is a \yes-instance and the latter a \no-instance.

\begin{figure}[t]
	\centering
	\begin{tikzpicture}[scale=0.8,every node/.style={scale=0.7}]
		\node[draw,fill=red,inner sep=3pt,circle] (root) at (10,10) {};
		
		\node[draw,fill=orange,inner sep=3pt,circle,yshift=5mm] (c3) [below= of root] {};
		
		\node[draw,fill=cyan,inner sep=3pt,circle,yshift=5mm,xshift=-5mm] (c31) [below= of c3] {};
		\node[draw,fill=yellow,inner sep=3pt,circle,yshift=5mm,xshift=5mm] (c32) [below= of c3] {};
		
		\draw[-{Stealth[length=6pt]}] (root) -> (c3);
		
		\draw[-{Stealth[length=6pt]}] (c3) -> (c31);
		\draw[-{Stealth[length=6pt]}] (c3) -> (c32);
		
		\node[xshift=14mm] [left= of c3] {$u'$};
		\node[xshift=14mm] [left= of c32] {$v'$};
	\end{tikzpicture}
	\begin{tikzpicture}[scale=0.8,every node/.style={scale=0.7}]
		\node[draw,fill=red,inner sep=3pt,circle] (root) at (10,10) {};
		
		\node[draw,fill=orange,inner sep=3pt,circle,yshift=5mm,xshift=-13mm] (c3) [below= of root] {};
		\node[draw,fill=orange,inner sep=3pt,circle,yshift=5mm,xshift=13mm] (c5) [below= of root] {};
		
		\node[draw,fill=cyan,inner sep=3pt,circle,yshift=5mm,xshift=-4mm] (c30) [below= of c3] {};
		\node[draw,fill=yellow,inner sep=3pt,circle,yshift=5mm,xshift=4mm] (c32) [below= of c3] {};
		
		\node[draw,fill=cyan,inner sep=3pt,circle,yshift=5mm,xshift=-4mm] (c51) [below= of c5] {};
		\node[draw,fill=yellow,inner sep=3pt,circle,yshift=5mm,xshift=4mm] (c52) [below= of c5] {};
		
		\draw[-{Stealth[length=6pt]}] (root) -> node[left] {1} (c3);
		\draw[-{Stealth[length=6pt]}] (root) -> node[right] {1} (c5);
		
		\draw[-{Stealth[length=6pt]}] (c3) -> node[left] {20} (c30);
		\draw[-{Stealth[length=6pt]}] (c3) -> node[right] {1} (c32);
		
		\draw[-{Stealth[length=6pt]}] (c5) -> node[left] {1} (c51);
		\draw[-{Stealth[length=6pt]}] (c5) -> node[right] {20} (c52);
	\end{tikzpicture}
	\begin{tikzpicture}[scale=0.8,every node/.style={scale=0.7}]
		\node[draw,fill=red,inner sep=3pt,circle] (root) at (10,10) {};
		
		\node[draw,fill=cyan,inner sep=3pt,circle,yshift=-7mm,xshift=-8mm] (c31) [below= of root] {};
		\node[draw,fill=yellow,inner sep=3pt,circle,yshift=-7mm,xshift=8mm] (c32) [below= of root] {};
		
		\draw[-{Stealth[length=6pt]},white] (8,10) -> (8,9);
		\draw[-{Stealth[length=6pt]}] (root) to[bend right=8] (c31);
		\draw[-{Stealth[length=6pt]}] (root) to[bend left=8] (c32);
		
	\end{tikzpicture}
	\begin{tikzpicture}[scale=0.8,every node/.style={scale=0.7}]
		\node[draw,fill=red,inner sep=3pt,circle] (root) at (10,10) {};
		
		\node[draw,fill=cyan,inner sep=3pt,circle,yshift=5mm,xshift=-4mm] (c30) [below= of c3] {};
		\node[draw,fill=yellow,inner sep=3pt,circle,yshift=5mm,xshift=4mm] (c32) [below= of c3] {};
		
		\node[draw,fill=cyan,inner sep=3pt,circle,yshift=5mm,xshift=-4mm] (c51) [below= of c5] {};
		\node[draw,fill=yellow,inner sep=3pt,circle,yshift=5mm,xshift=4mm] (c52) [below= of c5] {};
		
		\draw[-{Stealth[length=6pt]}] (root) -> node[left] {20} (c30);
		\draw[-{Stealth[length=6pt]}] (root) -> node[right] {2} (c32);
		
		\draw[-{Stealth[length=6pt]}] (root) -> node[left] {1} (c51);
		\draw[-{Stealth[length=6pt]}] (root) -> node[right] {21} (c52);
	\end{tikzpicture}
	\caption{An execution of Reduction Rule~6 in~\cite{PDD}.
		Left is the instance before the execution and right is the instance after the execution.
		In both instances, to the left is the pattern tree and on the right the phylogenetic tree.}
	\label{fig:mistake}
\end{figure}
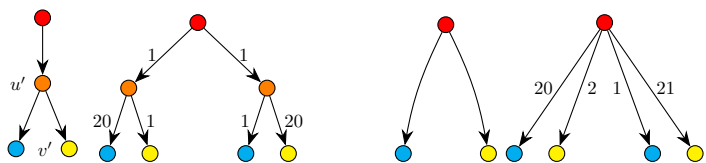%

Nevertheless, we still believe \PDD to be \FPT with respect to $k+h$ and
%
%
%
we further believe that successful attempts for \PDD can then be generalized to \MAPPDD when also the parameter~$\delta$ is considered.

\fi

\end{document}

%% file: commands.tex
\usepackage[utf8]{inputenc}
\usepackage{amsmath}
\usepackage{amssymb}
\usepackage{amsthm}
\usepackage{etex}
\usepackage{soul}
\usepackage{enumerate}
\usepackage{dsfont}
\usepackage{thmtools}
\usepackage{mathtools}
\usepackage{graphicx}
\usepackage{multirow}
\usepackage[table]{xcolor}
\usepackage{tikz}
\usetikzlibrary {arrows.meta} 
\usetikzlibrary{graphs}
\usepackage{placeins}
\usepackage{nicefrac}
\usepackage{tabularx}
\usepackage{algorithm2e}
\usepackage[T1]{fontenc}
\usepackage{etoolbox}
\usepackage{hyperref}
\usepackage{paralist}
\usepackage{booktabs}
\usepackage{bbding}
\hypersetup{
	colorlinks,
	linkcolor={red!50!black},
	citecolor={blue!90!black!80},
	urlcolor={blue!80!black}
}
\usepackage{cleveref}

\usepackage[disable]{todonotes}
\usepackage{lineno}
\nolinenumbers

\newtheorem{rr}{Reduction~Rule}{\upshape\itshape}{\upshape\rmfamily}

{\upshape\itshape}{\upshape\rmfamily}
{\upshape\itshape}{\upshape\rmfamily}
\usepackage{multirow}

\newcommand{\kommentar}[1]{}
\newcommand{\Oh}{\ensuremath{\mathcal{O}}}
\newcommand{\SSS}{\ensuremath{\mathcal{S}}}

\newcommand{\proofpara}[1]{\smallskip
	
	\noindent\textit{#1.}}

\DeclareMathOperator{\poly}{poly}
\DeclareMathOperator{\twwithoutN}{{tw}}
\DeclareMathOperator{\swwithoutN}{{sw}}
\DeclareMathOperator{\retwithoutN}{{ret}}
\DeclareMathOperator{\off}{off}

\DeclareMathOperator{\DP}{DP}

\newcommand{\Dbar}{{\ensuremath{\overline{D}}}\xspace}
\newcommand{\Hgt}{{\ensuremath{H}}\xspace}
\newcommand{\kbar}{{\ensuremath{\overline{k}}}\xspace}
\newcommand{\w}{{\ensuremath{\omega}}\xspace}
\newcommand\lb{\linebreak}

\newcommand\Recc[1]{Recurrence~(\ref{#1})}

\newcommand{\predators}[1]{{\ensuremath{\text{pred}(#1)}}\xspace}
\newcommand{\prey}[1]{{\ensuremath{\text{prey}(#1)}}\xspace}
\newcommand{\predatorsE}[1]{{\ensuremath{\text{pred}^{(E)}(#1)}}\xspace}
\newcommand{\preyE}[1]{{\ensuremath{\text{prey}^{(E)}(#1)}}\xspace}

\newcommand{\yes}{{\normalfont\texttt{yes}}\xspace}
\newcommand{\no}{{\normalfont\texttt{no}}\xspace}

\newcommand{\Wh}[1]{{\normalfont W[#1]}\xspace}
\newcommand{\NP}{{\normalfont{NP}}\xspace}

\newcommand{\FPT}{{\normalfont{FPT}}\xspace}

\newcommand{\Instance}{{\ensuremath{\mathcal{I}}}\xspace}
\newcommand{\Net}{{\ensuremath{\mathcal{N}}}\xspace}
\newcommand{\Tree}{{\ensuremath{\mathcal{T}}}\xspace}
\newcommand{\Food}{{\ensuremath{\mathcal{F}}}\xspace}
\newcommand{\ret}{{\ensuremath{\retwithoutN_{\Net}}}\xspace}
\newcommand{\tw}{{\ensuremath{\twwithoutN_{\Net}}}\xspace}
\newcommand{\sw}{{\ensuremath{\swwithoutN_{\Food}}}\xspace}
\newcommand{\PD}{\PDsub\Net}
\newcommand{\PDsub}[1]{{\ensuremath{PD_{#1}}}\xspace}

\newcommand{\problemdef}[3]{
	\begin{quote}
		\normalsize\textsc{#1} \smallskip \\
		\begin{tabularx}{0.92\textwidth}{@{}l@{\hspace{3pt}}X}
			\normalsize\textbf{Input:}    & \normalsize#2 \\
			\normalsize\textbf{Question:} & \normalsize#3
		\end{tabularx}
	\end{quote}
}

\newcommand{\PROB}[1]{{{\normalfont\textsc{#1}}}\xspace}

\newcommand{\MPD}{\PROB{Max-PD}}
\newcommand{\MPDlong}{\PROB{Maximize Phylogenetic Diversity}}